\documentclass[pra,nobalancelastpage,danish,twocolumn,nofootinbib,superscriptaddress]{revtex4-1}
\pdfoutput=1

\usepackage{amsmath}
\usepackage{amssymb}
\usepackage{graphicx}
\usepackage{dcolumn}
\usepackage{bm}

\usepackage[utf8]{inputenc}
\usepackage[T1]{fontenc}
\usepackage{mathptmx}
\usepackage{subeqnarray}
\usepackage[colorlinks=true,linkcolor=blue, citecolor=blue, urlcolor=blue, bookmarks]{hyperref}
\usepackage{amsthm}

\newtheorem{thm}{Theorem}
\newtheorem{prop}{Proposition}
\newtheorem{cor}{Corollary}

\theoremstyle{definition}
\newtheorem{defn}{Definition}[section]
\newtheorem{example}{Example}[section]

\theoremstyle{remark}
\newtheorem{rem}{Remark}

\newcommand{\ket}[1]{\vert#1\rangle}
\newcommand{\bra}[1]{\langle#1\vert}
\newcommand{\bracket}[2]{\langle#1|#2\rangle}
\newcommand{\tr}{\mathrm{Tr}}
\newcommand{\prob}{\mathrm{Prob}}
\newcommand{\ii}{\mathrm{i}}

\begin{document}

\title[On Possibilistic Conditions to Contextuality and Nonlocality]{On Possibilistic Conditions to Contextuality and Nonlocality}

\author{Leonardo Santos}
\email{leo\_vieira@usp.br}
\author{Barbara Amaral}%
\affiliation{ 
Department of Mathematical Physics, Institute of Physics, University of São Paulo, Rua do Matão 1371, São Paulo 05508-090, São Paulo, Brazil
}%

\date{\today}

\begin{abstract}
Contextuality and nonlocality are non-classical properties exhibited by quantum statistics whose implications profoundly impact both foundations and applications of quantum theory. In this paper we provide some insights into logical contextuality and inequality-free proofs. The former can be understood as the possibility version of contextuality, while the latter refers to proofs of quantum contextuality/nonlocality that are not based on violations of some noncontextuality (or Bell) inequality. The present work aims to build a bridge between these two concepts from what we call possibilistic paradoxes, which are sets of possibilistic conditions whose occurrence implies contextuality/nonlocality. As main result, we demonstrate the existence of possibilistic paradoxes whose occurrence is a necessary and sufficient condition for logical contextuality in a very important class of scenarios. Finally, we discuss some interesting consequences arising from the completeness of these possibilistic paradoxes.
\end{abstract}

\maketitle

\section{Introduction}\label{sec:1}

Contextuality is a property of non-classical statistics that refers to the impossibility of their reproduction by models in which measurement outcomes revel pre-existing system properties whose values are independent of which (or whether) other compatible measurements are jointly performed. When these statistics are obtained from scenarios with space-like separations, contextuality is called nonlocality. Since the seminal works by Bell \cite{Bell64,Bell66} and Kochen-Specker \cite{KS67}, it is known that quantum statistical predictions can be contextual/nonlocal. In fact, it is currently known that contextuality and nonlocality are fundamental features of quantum theory with important implications on foundations \cite{Amaral14,CSW10}, computation \cite{HWVE14,R13,DGBR14,FRB18}, and communication \cite{Ekert91,BHK05}.

The assumption of noncontextuality (NC) imposes strong constraints on the possible empirically observed probabilities. A powerful way to express such restrictions is through linear inequalities that are obeyed whenever a description in terms of NC models is possible. Such inequalities generalize the notion of Bell inequalities for nonlocality, and they are usually called noncontextuality inequalities (NCI) \cite{ATC18}. Currently, the best known instances of NCIs are those proposed by Clauser-Horne-Shimony-Holt (CHSH) \cite{CHSH69} and Klyachko-Can-Binicioglu-Shumovsky (KCBS) \cite{KCBS08}, since both play an almost paradigmatic role in quantum foundations and quantum information theory \cite{Scarani19,BCPSW14,ATC18,BCGKL}. Furthermore, any set of NCIs that defines an H-representation for the NC polytope provides also necessary and sufficient conditions for contextuality \cite{ATC18}.

A different and more intuitive approach to obtain contradictions between quantum predictions and NC models is by using the possibilistic information of quantum probabilistic data. In other words, only the information of which results are possible/impossible can be sufficient to demonstrate its contextual character. Such an approach gave rise to so-called "inequality-free proofs", pioneered by Heywood and Redhead \cite{HR83}; Greenberger, Horne, Shimony, and Zeilinger \cite{GHZ89,GHSZ90}; and Hardy \cite{Hardy92,Hardy93}. We highlight the latter, which stands out for its simplicity and generality, being considered the simplest or the best proof of quantum nonlocality \cite{Mermin95}. The fundamental Hardy's ideas were to demonstrate how a small set of possibilistic conditions implies nonlocality, and then explicit show how such conditions can be realized within a quantum mechanical system.

Several generalizations and developments of Hardy's ideas have been proposed in the last decades \cite{GK98,Cereceda04,CCXSWK13,ACY16,MZXSGYC18}. In particular, we highlight the contribution due to Cabello \textit{et al.} \cite{CBTCB13}. In that paper, the authors proposed a "Hardy-like proof" to quantum contextuality in a scenario without space-like separation, namely the KCBS scenario \cite{CBTCB13,KCBS08}. In addition to theoretical developments, inequality-free proofs for quantum contextuality and nonlocality have been subject to several experimental verification \cite{BBMH97,BCMM05,YMZXXSCXLG19,MANCB14,CZWWFK17}.

A formalism to deal with contextuality and nonlocality in a unified framework was proposed by Abramsky and Brandenburger in Ref. \cite{AB11}. The formalism is mainly based on a generalization of Fine's theorem \cite{Fine82}, in such a way that contextuality means the impossibility of consistent global descriptions of probabilistic data. More precisely, using some insights provided by sheaf theory, the authors state contextuality as a phenomenon that arise whenever local consistence does not imply global consistence in a set of distributions (that is, normalized functions from a nonempty set to a commutative semiring). The usual notion is recovered when sets of probability distributions are considered. When sets of possibility distributions are considered (\textit{i.e.} when the semiring is the Boolean), we deal with the concept of logical contextuality (LC). In the present paper, this concept will have a central role, since it provides a very systematic and general way of defining (non)contextuality from the possibilistic structure of probabilistic data.

Both inequality-free proofs and LC are related with violations on logical constrains imposed by the assumption of NC. The connection between these two concepts, on the other hand, is not yet so clear. In an inequality-free proof we have a set of possibilistic conditions that implies contextuality. In the present paper, we call any of these sets of conditions a possibilistic paradox (PP). It is clear that the occurrence of some PP implies LC, suggesting that these are the possibilistic analogue of NCI. To make this analogy more accurate and for a better understanding of LC, however, some questions still need to be addressed. The main is the following: what are (if any) set(s) of PP(s) whose occurrence is a necessary and sufficient condition for LC? A partial answer to this question was proposed by Mansfield and Fritz \cite{MF12} for some bipartite Bell scenarios from a generalization of the PP proposed in the Hardy's proof \cite{Hardy92,Hardy93}.

In the present paper we generalize the results of Ref. \cite{MF12}. That is, we demonstrate that the occurrence of a specific set of PPs is a necessary and sufficient condition for LC in an important class of scenarios, namely the scenarios whose contexts have at most two dichotomic measurements. Furthermore, we discuss some interesting consequences arising from the completeness of the proposed sets of PPs. The PPs we propose have a very clear and friendly form, and they are defined as $n$-cycle generalizations of those present in the Hardy \cite{Hardy92,Hardy93} and Cabello \textit{et al.} \cite{CBTCB13} inequality-free proofs. As a consequence of this result, we conclude that the only strongly contextual behaviors for the $n$-cycle scenarios are analogously to the Popescu-Rohrlich boxes. Moreover, with a minor modification of the proposed PPs, its completeness still holds for the $n$-cycle scenarios with non-dichotomic measurements. Finally, we show that for all $n\geq 4$ it is possible to construct a "Hardy-like proof" to quantum contextuality in the $n$-cycle scenarios.

We have organized the paper as follows: in Sec. \ref{sec:II} we review the definitions of compatibility scenarios, behaviors, and contextuality following the compatibility hypergraph approach; in
Sec. \ref{sec:III} we present the notion of logical contextuality; in Sec. \ref{sec:IV} we present the our results. We finish the work with a discussion in Sec. \ref{sec:V}.

\section{Contextuality: Compatibility Hypergraph Approach}\label{sec:II}

Suppose we have a hypothetical physical system on which a finite set of measurements can be performed. Every available measurement can be labeled by an element in a finite set $X$, and when it is performed a macroscopic effect is observed. Consider that there are only a finite number of macroscopically distinguishable outcomes, in such a way that we can label them by elements in a finite set $O$. Moreover, some sets of measurements can be jointly performed (called compatible sets), while others cannot (called incompatible sets). The information of which measurements are available, its possible outcomes, and which compatible sets we choose to jointly perform defines a compatibility scenario and its associated compatibility hypergraph \cite{ATC18}. 

\begin{defn}\label{defn:II.1}
A \textit{compatibility scenario} (or just a scenario) is defined by triple $(X,\mathcal{C},O)$, where $X$ and $O$ are finite sets, and $\mathcal{C}$ is a family of subsets of $X$ such that $\cup \mathcal{C}=X$, and $C\subseteq C'$ implies $C=C'$ whenever $C,C'\in\mathcal{C}$. The \textit{compatibility hypergraph} of a scenario is the hypergraph whose set of vertices is $X$ and the set of hyperedges is $\mathcal{C}$.
\end{defn}

Given a compatibility scenario $(X,\mathcal{C},O)$, each element of $\mathcal{C}$ defines a maximal set of compatible measurements, called a \textit{context}. The family $\mathcal{C}$ is usually called the \textit{compatibility cover} of the scenario. Joint outcomes for the measurements in a context $C$ can be represented by functions from $C$ to $O$, and the set of these functions is denoted by $O^C$. Since $O^C\simeq O^{|C|}$, joint outcomes can also be represented by strings with $|C|$ elements, one for each element of $C$, that is $(s(m):m\in C)$. In what follows, it will be useful to keep both representations in mind.

Given a joint outcome $s\in O^C$, if we want to refer to a specific information concerning a subset of measurements $U\subseteq C$, then we just need to restrict the function $s$ to the subset $U$. Following the notation adopted in Ref. \cite{AB11}, we denote this restriction by $s|_U$. Intuitively, this operation transforms the string $(s(m):m\in C)$ into $(s(m):m\in U)$.

\begin{rem}
In what follows, if $(X,\mathcal{C},O)$ is a compatibility scenario and $\Omega\subseteq X$ is non-empty, $O^\Omega$ will always denote the set of functions from $\Omega$ to $O$.
\end{rem}

\begin{figure}
\centering
\includegraphics[angle=0,scale=0.5]{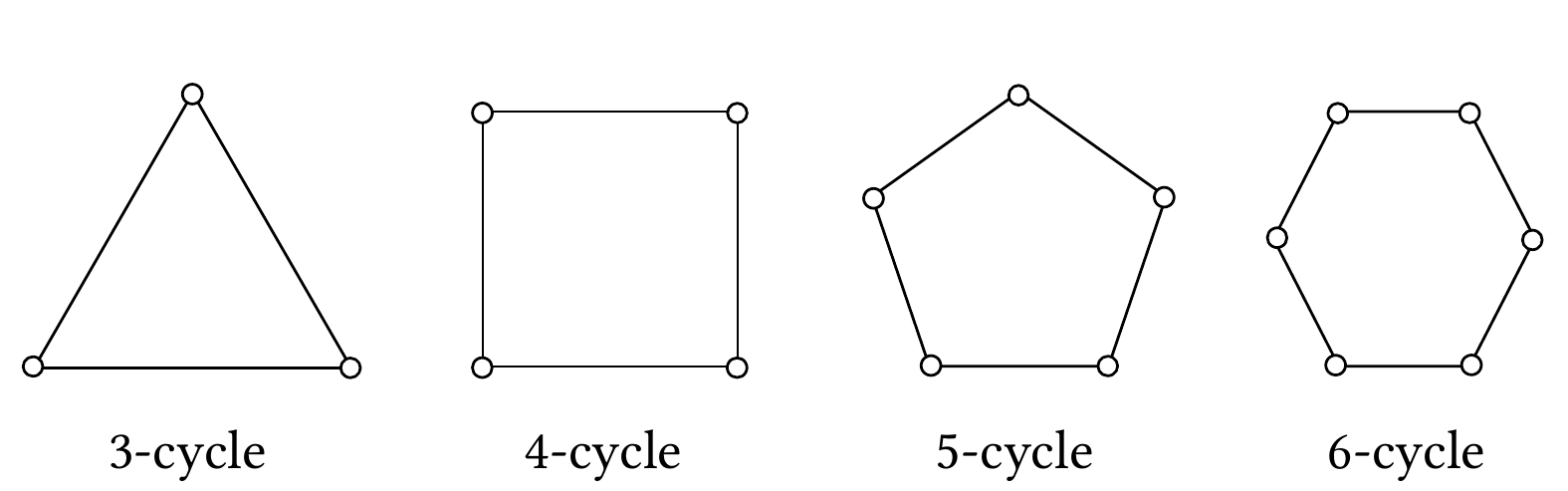}
\caption{Compatibility graph for $n$-cycle scenarios [Example \ref{example:II.1}].}
\label{fig:1}
\end{figure}
\begin{figure}
\centering
\includegraphics[angle=0,scale=0.5]{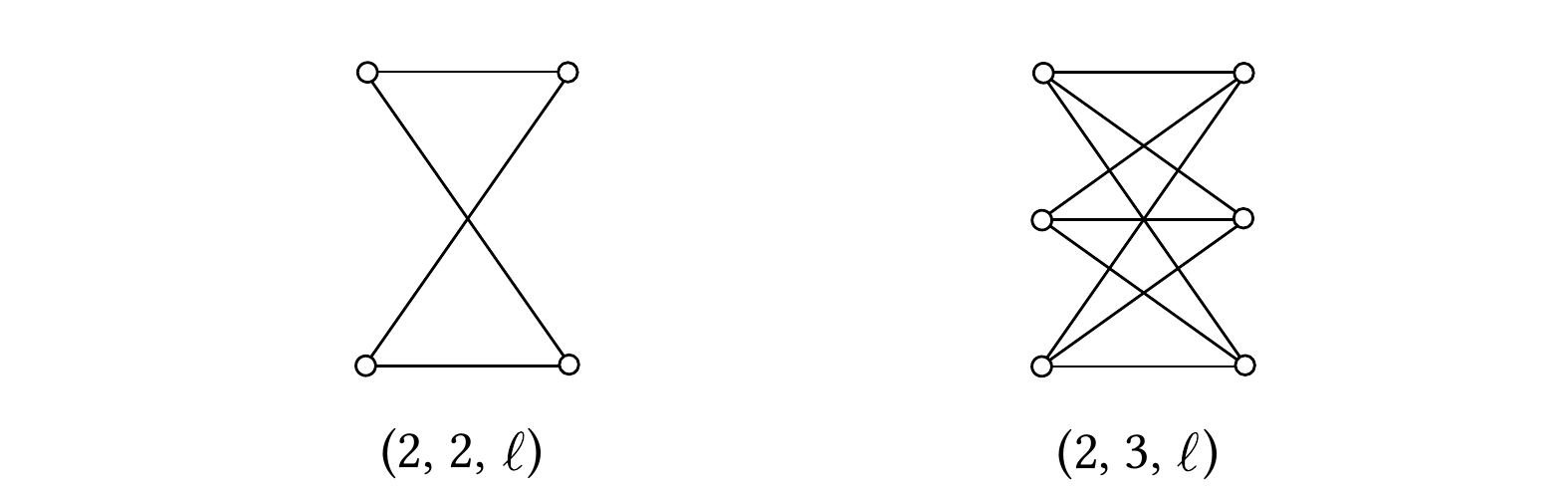}
\caption{Compatibility graph for bipartite Bell-type scenarios $(2,k,\ell)$ [Example \ref{example:II.2}].}
\label{fig:2}
\end{figure}

\begin{example}[$n$-cycle scenarios]\label{example:II.1}
The $n$-cycle scenarios are defined by $n\geq 3$ dichotomic measurements, $X:=\{M_1,\dots,M_n\}$, $O:=\{0,1\}$, with contexts $C_i:=\{M_i,M_{i+1}\}$, where the sum in index is taken modulo $n$. For $n=4$ and $n=5$ we have, respectively, the well known CHSH and KCBS scenarios \cite{CHSH69,KCBS08}.
\end{example}

\begin{example}[Bell scenarios]\label{example:II.2}
A Bell scenario is defined by three natural numbers $(n,k,\ell)$, where $n\geq 2$ is the number of spatially separated parts, $k\geq2$ is the number of available measurement of each part, and $\ell\geq2$ is the number of outcomes for each measurement. The contexts are constructed taking exactly one measurement of each part.
\end{example}

A compatibility scenario is said to be \textit{simple} if its contexts have at most two measurements, \textit{e.g.} bipartite Bell scenarios $(2,k,\ell)$ and $n$-cycle. In this case, the associated compatibility hypergraph is just a graph (see Fig. \ref{fig:1} and Fig. \ref{fig:2}). If the compatibility hypergraph of a scenario has no cycles as induced subgraphs we then say that the scenario is \textit{acyclic}.

A single measurement event in general does not provide sufficient information to compare experimental results with theoretical predictions. In fact, a measurement event provides one (and only one) joint outcome, while (in general) we "only" have probability distributions. The experimenter needs to repeat the same experimental procedure (preparation followed by measurements) a number of times in order to estimate probabilities by relative frequencies.

Given a compatibility scenario $(X,\mathcal{C},O)$, we say that $p$ is a probability distribution for the context $C\in\mathcal{C}$ if it is a normalized function from the set of joint outcomes $O^C$ to the non-negative real numbers, that is, $p:O^C\to \mathbb{R}^+$ such that
\begin{eqnarray}\label{eq:1}
\sum_{s\in O^C} p(s)=1.
\end{eqnarray}
We denote by $\prob(O^C)$ the set of probability distributions for the context $C\in \mathcal{C}$. If $U\subseteq C$ is a subset of measurements of the context $C$, we then define the marginal probability distribution for $U$ given $p\in \prob(O^C)$ as
\begin{eqnarray}\label{eq:2}
p|_U(u):=\sum_{s\in O^C \hspace{0.05cm}:\hspace{0.05cm} s|_U=u} p(s),
\end{eqnarray}
for all $u\in O^U$. That is, the summation in the above equation is over the joint outcomes $s\in O^C$ for the context $C$ compatible with a given information provided in the subset $U\subseteq C$ by $u\in O^U$, that is $s|_U=u$.

The probabilistic information obtained from a physical system prepared in some way given a compatibility scenario defines what we call a \textit{behavior}.

\begin{defn}\label{defn:II.2}
A behavior for a scenario $(X,\mathcal{C},O)$ is defined by a set of probability distribution over $O^C$, one for each $C\in \mathcal{C}$, \textit{i.e.}: $\mathrm{B}=\{p_C\in \prob(O^C):C\in \mathcal{C}\}$.
\end{defn}

The behaviors we will consider obey a generalized non-signaling condition, which we call \textit{nondisturbance condition} (ND).

\begin{defn}\label{defn:II.3}
A behavior $\mathrm{B}=\{p_C\in \prob(O^C):C\in \mathcal{C}\}$ for a scenario $(X,\mathcal{C},O)$ obeys the nondisturbance condition if marginal distributions agree on overlapping contexts, \textit{i.e.}
\begin{eqnarray}\label{eq:3}
\sum_{s\in O^C \hspace{0.05cm}:\hspace{0.05cm} s|_{C\cap C'}=u}p_C(s)=\sum_{s'\in O^{C} \hspace{0.05cm}:\hspace{0.05cm} s'|_{C\cap C'}=u} p_{C'}(s'), 
\end{eqnarray}
for all $C,C'\in \mathcal{C}$ and $u\in O^{C\cap C'}$.
\end{defn}

Another important behaviors set is the \textit{noncontextual} (NC) one.

\begin{defn}\label{defn:II.4}
A behavior $\mathrm{B}=\{p_C\in \prob(O^C):C\in \mathcal{C}\}$ for a scenario $(X,\mathcal{C},O)$ is noncontextual if there exists a probability distribution $p$ over $O^X$ reproducing $\mathrm{B}$ by marginals, that is
\begin{eqnarray}\label{eq:4}
p|_C(s):=\sum_{t\in O^X \hspace{0.05cm}:\hspace{0.05cm} t|_C=s} p(t)=p_C(s),
\end{eqnarray}
for all $C\in \mathcal{C}$ and $s\in O^C$.
\end{defn}

It is easy to verify that NC behaviors are also ND. However, the converse in general does not hold: there are ND behaviors that are not NC. When it happens we say that the behavior is \textit{contextual}.

\begin{defn}\label{defn:II.5}
A ND behavior $\mathrm{B}$ for a scenario $(X,\mathcal{C},O)$ is contextual if there is no probability distribution over $O^X$ such that the condition (\ref{eq:2}) hold for all $C\in \mathcal{C}$ and $s\in O^C$. 
\end{defn}

\begin{rem}
When dealing with Bell scenarios (or their possible generalizations) contextuality is called \textit{nonlocality}.
\end{rem}

The definition \ref{defn:II.5} states contextuality as a phenomenon that arise when local consistence (ND) does not imply global consistence (NC) in a set of probability distributions. Although it seems quite abstract, definition \ref{defn:II.5} is equivalent to the usual Bell-Kochen-Specker notion, which refers to the impossibility of reproducing the statistics of a given behavior by NC hidden-variable models. This equivalence is the content of Fine-Abramsky-Brandenburger theorem \cite{AB11,Fine82}. As an immediate consequence of this result, it can be concluded that a given behavior is NC if, and only if, it can be classically realized, \textit{i.e.} it is compatible with the classical notion of probabilities given by the Kolmogorovian theory with a unique probability space \cite{ATC18}.

Among all possible behaviors realized within non-classical theories, we are particularly interested in the quantum ones.

\begin{defn}
A \textit{quantum model} for a behavior $\mathrm{B}=\{p_C\in\prob(O^C):C\in\mathcal{C}\}$ for a scenario $(X,\mathcal{C},O)$ is defined by: (\textit{i}) a pair $(\rho,\mathcal{H})$, where $\mathcal{H}$ is a (complex and separable) Hilbert space and $\rho$ is a density operator on $\mathcal{H}$; (\textit{ii}) for each $m\in X$, a partition of the identity $\mathbb{I}$ of $\mathcal{H}$ into $|O|$ orthogonal projectors, \textit{i.e.} a collection of orthogonal projectors $\{\mathrm{P}_{m;\hspace{0.05cm}o}:o\in O\}$ such that
\begin{eqnarray*}
\sum_{o\in O}\mathrm{P}_{m;\hspace{0.05cm}o}=\mathbb{I},
\end{eqnarray*}
and $\mathrm{P}_{m;\hspace{0.05cm}o}\mathrm{P}_{m;\hspace{0.05cm}\tilde{o}}=0$ whenever $o\neq \tilde{o}$. We require that $[\mathrm{P}_{m;\hspace{0.05cm}o},\mathrm{P}_{\tilde{m};\hspace{0.05cm}\tilde{o}}]=0$, for all $o,\tilde{o}\in O$, whenever $m$ and $\tilde{m}$ belongs to the same context. The probability of obtain a joint outcome $s\in O^C$ will be
\begin{eqnarray}\label{eq:5}
p_C(s)=\tr\Big(\prod_{m\in C}\mathrm{P}_{m;\hspace{0.05cm}s(m)}\hspace{0.1cm}\rho\Big)
\end{eqnarray}
If the probabilities of a behavior $\mathrm{B}$ can be obtained by a quantum model, we then say that $\mathrm{B}$ is a quantum behaviour.
\end{defn}

It is a well-known fact that quantum behaviors can display contextuality \cite{BCGKL}. A very useful way of witnessing contextuality in a given behavior is through the violation of some NC inequality of the scenario. Such inequalities arise naturally from the geometry of the set of NC behaviors, in particular from the fact that such a set is a convex polytope \cite{ATC18}. In literature there is a large number of demonstrations for contextuality in quantum behaviors (or quantum contextuality for short) based on the violation of some NC inequality \cite{BCGKL}. In this paper we highlight the $n$-cycle inequalities \cite{AQBTCC13}, which are NC inequalities for the $n$-cycle scenarios defined by
\begin{eqnarray}\label{eq:6}
\sum_{i=1}^n \gamma_i \langle M_i M_{i+1}\rangle \leq n-2,
\end{eqnarray}
where $\langle M_i M_{i+1}\rangle=2(p_i(0,0)+p_i(1,1))-1$, $p_i$ denotes the probability distribution associated to the context $\{M_i,M_{i+1}\}$, $\gamma_i\in \{-1,+1\}$, and the number of $\gamma_i$ equal to $-1$ is odd. The inequalities in Eq. (\ref{eq:6}) also defines the facets of the full-dimensional NC polytope of the $n$-cycle scenario, and then their violations is a necessary and sufficient condition for a behavior to be contextual. Furthermore, it is possible to demonstrate that quantum behaviors can violate such inequalities with maximum value given by \cite{AQBTCC13}
\begin{eqnarray}\label{eq:7}
\begin{cases}
\displaystyle{\frac{3n\cos(\pi/n)-n}{1+\cos(\pi/n)}} \hspace{0.5cm} \textrm{if $n$ is odd} \\ 
n\cos(\pi/n) \hspace{1.325cm} \textrm{if $n$ is even}
\end{cases}.
\end{eqnarray}

We have been highlighting the $n$-cycle scenarios throughout our discussion for two many reasons. First, such scenarios are the simplest ones for which contextual behaviors can be defined. In fact, Vorob'ev theorem \cite{Vorobev63,Vorobyev67} states that it is possible to define contextual behaviors in a given scenario if, and only if, its compatibility hypergraph is not acyclic \cite{BMC16}. The second reason is that such scenarios will be on the ground of our main results.

\section{Logical Contextuality}\label{sec:III}

Herein we aim to investigate global inconsistencies that can arise from the possibilistic structure of probabilistic data. In other words, we are interested in the possibilistic version of contextuality, called \textit{logical contextuality} (LC). This type of nonclassicality was precisely formulated in Ref. \cite{AB11}. Since then, many efforts have been done in order to understand it \cite{Mansfield17,SWPBR17,Mansfield13,Constantinthesis,Simmonsthesis,Caruthesis}. In the present section we put in forward the mathematical theory of LC in a presentation that emphasizes its similarities and differences with its probabilistic analogue.

First of all, let us discuss how to consistently describe the possibilistic structure of probabilistic data. In particular we aim to precisely define how properties of (probabilistic) behaviors, such as nondisturbance and (non)contextuality, can be translated to a possibilistic language. 

Given a compatibility scenario $(X,\mathcal{C},O)$, we call $p$ a \textit{possibility distribution} for the context $C\in\mathcal{C}$ if $p$ is a normalized function from the set of joint outcomes $O^C$ to the Booleans $\mathbb{B}=(\{0,1\},\lor,\land,\lnot)$, that is $p:O^C\to \mathbb{B}$ such that
\begin{eqnarray}\label{eq:8}
\bigvee_{s\in O^C}p(s)=1.
\end{eqnarray}
As in the probabilistic case, if $U\subseteq C$ is a subset of measurements of the context $C$, we then define the marginal possibility distribution for $U$ given $p$ as
\begin{eqnarray}\label{eq:9}
p|_U(u):=\bigvee_{s\in O^C \hspace{0.05cm}:\hspace{0.05cm} s|_U=u} p(s),
\end{eqnarray}
for all $u\in O^U$. With this definition, we define the possibilistic structure of a behavior, or its possibilistic collapse \cite{AB11}, as follows.

\begin{defn}\label{defn:III.2}
The \textit{possibilistic collapse} of a behavior $\mathrm{B}=\{p_C\in \prob(O^C):C\in\mathcal{C}\}$ for a scenario $(X,\mathcal{C},O)$ is defined by the set of possibility distributions $\bar{p}_C:O^C\to \mathbb{B}$ defined by the rule
\begin{equation}\label{eq:10}
p_C(s)\mapsto \bar{p}_C(s)=\begin{cases} 1 \hspace{0.5cm} \textrm{if } p_C(s)>0 \\ 0\hspace{0.5cm} \textrm{otherwise} \end{cases}.
\end{equation}
\end{defn}

\begin{rem}
In what follows we always denote by $\bar{p}$ the possibility distribution generated by $p$ by using the rule (\ref{eq:10}).
\end{rem}

The normalization condition for the probability distributions $p_C$ implies that each $\bar{p}_C$ is a well-defined possibility distribution over $O^C$. The possibilistic collapse turns a set of probability distributions into a set of possibility distributions. In other words, it provides the possibilistic structure of the behavior. An important feature of this operation is that it "preserves" the properties discussed in the previous section, namely NC and ND. In fact, this is the content of the following propositions (whose proofs will be omitted because they are trivial).  

\begin{prop}\label{prop:1}
If a behavior $\mathrm{B}=\{p_C\in \prob(O^C):C\in \mathcal{C}\}$ for a scenario $(X,\mathcal{C},O)$ is ND then 
\begin{equation}\label{eq:11}
\bigvee_{s\in O^C \hspace{0.05cm}:\hspace{0.05cm} s|_{C\cap C'}=t} \bar{p}_C(s)=\bigvee_{s'\in O^{C'} \hspace{0.05cm}:\hspace{0.05cm} s'|_{C\cap C'}=t}\bar{p}_{C'}(s'),
\end{equation}
for all $t\in O^{C\cap C'}$ and $C,C'\in\mathcal{C}$, which is the possibilistic version of ND condition (\ref{eq:1}).
\end{prop}

\begin{prop}\label{prop:2}
If a behavior $\mathrm{B}=\{p_C\in \prob(O^C):C\in \mathcal{C}\}$ for the scenario $(X,\mathcal{C},O)$ is noncontextual, then there exists a possibility distribution $\bar{p}$ over $O^X$ such that
\begin{eqnarray}\label{eq:12}
\bar{p}|_C(s):=\bigvee_{t\in O^X \hspace{0.05cm}:\hspace{0.05cm} t|_C=s} \bar{p}(t)=\bar{p}_C(s),
\end{eqnarray}
for all $C\in \mathcal{C}$ and $s\in O^C$.
\end{prop}

Once we have a well-defined notion of ND and global possibility distributions, the definition of (non)contextuality in possibilistic sense follows immediately.

\begin{defn}\label{definition:III.3}
A behavior $\mathrm{B}=\{p_C\in \prob(O^C):C\in \mathcal{C}\}$ for a scenario $(X,\mathcal{C},O)$ is \textit{logically noncontextual} if there exists a possibility distribution $\bar{p}$ over $O^X$ reproducing the possibilistic collapse of $\mathrm{B}$ by marginals, that is
\begin{eqnarray}\label{eq:13}
\bar{p}|_C(s):=\bigvee_{t\in O^X \hspace{0.05cm}:\hspace{0.05cm} t|_C=s} \bar{p}(t)=\bar{p}_C(s),
\end{eqnarray}
for all $C\in \mathcal{C}$ and $s\in O^C$. Otherwise, $\mathrm{B}$ is said to be \textit{logically contextual} (LC).
\end{defn}

Proposition \ref{prop:2} states that if a given behavior is NC in probabilistic sense then it is NC in possibilistic sense. However, it should be stressed that the converse in general does not hold.

An alternative and useful characterization of LC is given by the following theorem. 

\begin{thm}\label{thm:1}
A behavior $\mathrm{B}=\{p_C\in \prob(O^C):C\in\mathcal{C}\}$ for the scenario $(X,\mathcal{C},O)$ is LC if, and only if, there exists $s\in O^C$, for some $C\in\mathcal{C}$, such that
\begin{enumerate}
\item $p_C(s)>0$; 
\item if $t\in \{g\in O^X: g|_C=s\}$ then $p_{C'}(t|_{C'})=0$ for some $C'\in \mathcal{C}\setminus\{C\}$.
\end{enumerate}
\end{thm}
\begin{proof} 
Let $\mathrm{B}=\{p_C\in\prob(O^C):C\in\mathcal{C}\}$ be a behavior for the scenario $(X,\mathcal{C},O)$. First, suppose the existence of $s\in O^C$, for some $C\in \mathcal{C}$, such that the conditions 1 and 2 of the theorem hold. If $\bar{p}$ is a possibility distribution over $O^X$ reproducing $\mathrm{B}$ by marginals, then condition 2 implies that $\bar{p}(t)=0$ for all $t\in\{g\in O^X:g|_C=s\}$. However, this implies that $\bar{p}_C(s)=0$, which contradicts the condition 1. Therefore, $\mathrm{B}$ is LC.

In order to prove the converse, consider the contrapositive of the statement, \textit{i.e.} let us prove that the non-existence of such $s$ implies logical noncontextuality. Assuming that there is no $s\in O^C$ such that the conditions of the theorem hold, then $\mathrm{Supp}(\mathrm{B})
:=\{t\in O^X: p_C(t|_C)>0, \hspace{0.1cm} \forall \hspace{0.1cm} C\in\mathcal{C}\}$ is non-empty. Let $\bar{p}$ be the possibility distribution over $O^X$ such that $\bar{p}(t)=1$ if $t\in \mathrm{Supp}(\mathrm{B})$, and $\bar{p}(t)=0$ otherwise. From our assumption, for all $s\in O^C$ such that $p_C(s)>0$, there exists $t\in \mathrm{Supp}(\mathrm{B})$ such that $t|_C=s$. Thus,
\begin{equation}\label{eq:14}
\bar{p}|_C(s):=\bigvee_{t\in O^X \hspace{0.05cm}:\hspace{0.05cm} t|_C=s} \bar{p}(t)=\bigvee_{t\in \mathrm{Supp}(\mathrm{B}) \hspace{0.05cm}:\hspace{0.05cm} t|_C=s}\bar{p}(t)=1,
\end{equation}
hence $\bar{p}|_C(s)=\bar{p}_C(s)$ if $p_C(s)>0$. To the case in which $p_C(s)=0$, there is no $t\in \mathrm{Supp}(\mathrm{B})$ such that $t|_C=s$. Thus,
\begin{equation}\label{eq:15}
\bar{p}|_C(s):=\bigvee_{t\in O^X \hspace{0.05cm}:\hspace{0.05cm} t|_C=s} \bar{p}(t)=\bigvee_{t\in O^X\setminus\mathrm{Supp}(\mathrm{B})\hspace{0.05cm}:\hspace{0.05cm} t|_C=s}\bar{p}(t)=0,
\end{equation}
hence $\bar{p}|_C(s)=\bar{p}_C(s)$. Therefore, we conclude that $\bar{p}$ is a possibility distribution over $O^X$ such that $\bar{p}|_C(s)=\bar{p}_C(s)$ for all $C\in \mathcal{C}$ and $s\in O^C$, \textit{i.e.} $\mathrm{B}$ is logically noncontextual.
\end{proof}

Theorem \ref{thm:1} states that logical contextuality prohibits certain pre-measurement assignments to the values of all measurable properties at once. In fact, from condition 2 of theorem \ref{thm:1}, any noncontextual assignment to all measurement results would imply that $p_C(s)=0$. Hence, it is not possible to give a noncontextual assignment to all measurements results in such a way that the joint outcome $s$ is assigned to the measurements in context $C$, even though it has a non-zero probability of being observed. With this picture in mind, the following definitions will be useful.

\begin{defn}\label{definition:III.4}
Let $\mathrm{B}=\{p_C\in \prob(O^C):C\in \mathcal{C}\}$ be a behavior for a scenario $(X,\mathcal{C},O)$. The \textit{support of }$\mathrm{B}$ is a subset of $O^X$ defined by $\mathrm{Supp}(\mathrm{B}):=\{t\in O^X: p_C(t|_C)>0,\hspace{0.1cm} \forall\hspace{0.1cm} C\in \mathcal{C}\}$. The elements in $\mathrm{Supp}(\mathrm{B})$ will be called \textit{global assignments}. 
\end{defn}

By using the terms of the definition above, logical contextuality means the existence of some joint outcome with larger-than-zero probability that cannot be obtained by the restriction of any global assignment. The extreme case occurs when no global assignment is allowed, \textit{i.e.} when for all context $C\in \mathcal{C}$ there exists $s\in O^C$ such that the conditions 1 and 2 of the theorem \ref{thm:1} hold. A behavior with this property is said to be \textit{strongly contextual} (SC).

\begin{table}[!htb]
\begin{ruledtabular}
\begin{tabular}{ccccc}
&$(0,0)$&$(1,0)$&$(0,1)$&$(1,1)$\\
\hline
$\{A_1,B_1\}$ & $1$ & $0$ & $0$ & $1$ \\
$\{B_1,A_2\}$ & $1$ & $1$ & $1$ & $1$\\
$\{A_2,B_2\}$ & $1$ & $1$ & $1$ & $1$ \\
$\{B_2,A_1\}$ & $1$ & $1$ & $1$ & $1$ \\
\end{tabular}
\end{ruledtabular}
\caption{Possibilistic collapse of Bell's model [Example \ref{example:III.1}].}
\label{table:1}
\end{table}
\begin{table}[!htb]
\begin{ruledtabular}
\begin{tabular}{ccccc}
&$(0,0)$&$(1,0)$&$(0,1)$&$(1,1)$\\
\hline
$\{A_1,B_1\}$ & $1$ & $1$ & $1$ & $1$ \\
$\{B_1,A_2\}$ & $1$ & $0$ & $1$ & $1$\\
$\{A_2,B_2\}$ & $1$ & $1$ & $1$ & $0$ \\
$\{B_2,A_1\}$ & $1$ & $1$ & $0$ & $1$ \\
\end{tabular}
\end{ruledtabular}
\caption{Possibilistic collapse of Hardy's model [Example \ref{example:III.2}].}
\label{table:2}
\end{table}
\begin{table}[!htb]
\begin{ruledtabular}
\begin{tabular}{ccccc}
&$(0,0)$&$(1,0)$&$(0,1)$&$(1,1)$\\
\hline
$\{A_1,B_1\}$ & $1$ & $0$ & $0$ & $1$ \\
$\{B_1,A_2\}$ & $1$ & $0$ & $0$ & $1$\\
$\{A_2,B_2\}$ & $1$ & $0$ & $0$ & $1$ \\
$\{B_2,A_1\}$ & $0$ & $1$ & $1$ & $0$ \\
\end{tabular}
\end{ruledtabular}
\caption{Possibilistic collapse of PR-Box [Example \ref{example:III.3}].}
\label{table:3}
\end{table}

\begin{defn}\label{defn:III.5}
A behavior $\mathrm{B}=\{p_C\in \prob(O^C):C\in \mathcal{C}\}$ for a scenario $(X,\mathcal{C},O)$ is said to be \textit{strongly contextual} if $\mathrm{Supp}(\mathrm{B})=\emptyset$.
\end{defn}

In a sense, SC is the most extreme form of non-classicality exhibited by a behavior. In fact, as demonstrated by Abramsky and Brandenburger \cite{AB11}, all ND behavior can be written as a convex sum of a NC and a SC behaviors. This decomposition give rise to the so-called \textit{contextual fraction}, which is a measure of how contextual a behavior is \cite{AB11,ABM17}. It is possible to demonstrate that the contextual fraction is equal to one only for behaviors that exhibit SC \cite{AB11,ABM17}. Furthermore, such a contextuality measure plays a very important role in resource theories for contextuality, due to the fact that it is a monotone under all linear operations that preserve the NC set \cite{Amaral19}.

In order to illustrate these concepts and results, consider the "canonical" examples \cite{AB11,Mansfield13}: Bell's model \cite{Bell64}, Hardy's model \cite{Hardy93}, and the Popescu-Rohrlich box (PR-Box) \cite{PR94}. These are behaviors for the CHSH (4-cycle) scenario (or equivalently for the Bell-type scenario $(2,2,2)$). Each of these behaviors exhibit different "levels" of contextuality, namely: Bell's model is contextual but not LC; Hardy's model is LC but not SC; and PR-Box is SC. That is, we have the following qualitative hierarchy:
\begin{eqnarray*}
\textrm{Bell}<\textrm{Hardy}<\textrm{PR-Box}.
\end{eqnarray*}

\begin{example}[Bell's Model]\label{example:III.1} Proposed by Bell  \cite{Bell64} to demonstrate that quantum mechanics predicts nonlocal correlations, Bell's model has the possibilistic structure illustrated in Table \ref{table:1}.
\end{example}

\begin{example}[Hardy's Model]\label{example:III.2} Proposed by Hardy \cite{Hardy93} in its inequality-free proof of Bell's theorem, the possibilistic structure of Hardy's model is illustrated in Table \ref{table:2}.
\end{example}

\begin{example}[PR-Box]\label{example:III.3} Proposed by Popescu and Rohrlich \cite{PR94} to illustrate the existence of stronger than quantum non-signalizing correlations, the possibilistic collapse of PR-Box is illustrated in Table \ref{table:3}.
\end{example}

The possibilistic structure of a behavior can be visualized in its \textit{bundle diagram} \cite{BO18}. These consists of a base space formed by the compatibility hypergraph of the scenario (see Fig. \ref{fig:3} (a)) and, on top of each vertex, we added a fiber with the possible outcomes for the corresponding measurement (see Fig. \ref{fig:3} (b)). We connect two outcomes by an edge if the corresponding joint outcome is possible (\textit{i.e.} has non-zero probability). In this diagrammatic representation, global assignments are associated to loops, \textit{i.e.} closed paths traversing all the fibers exactly once; LC means the existence of some edge that does not belong to any loop; and SC means that no loop is allowed (see Fig. \ref{fig:4}). In Bell's model (Fig. \ref{fig:4} (a)) all edges belongs to at least one loop; in Hardy's model (Fig. \ref{fig:4} (b)) one can see that the red edge does not belong to any loop; and in the PR-Box (Fig. \ref{fig:4} (c)) there is no loop.

\begin{figure}
\centering
\includegraphics[scale=0.35,angle=0]{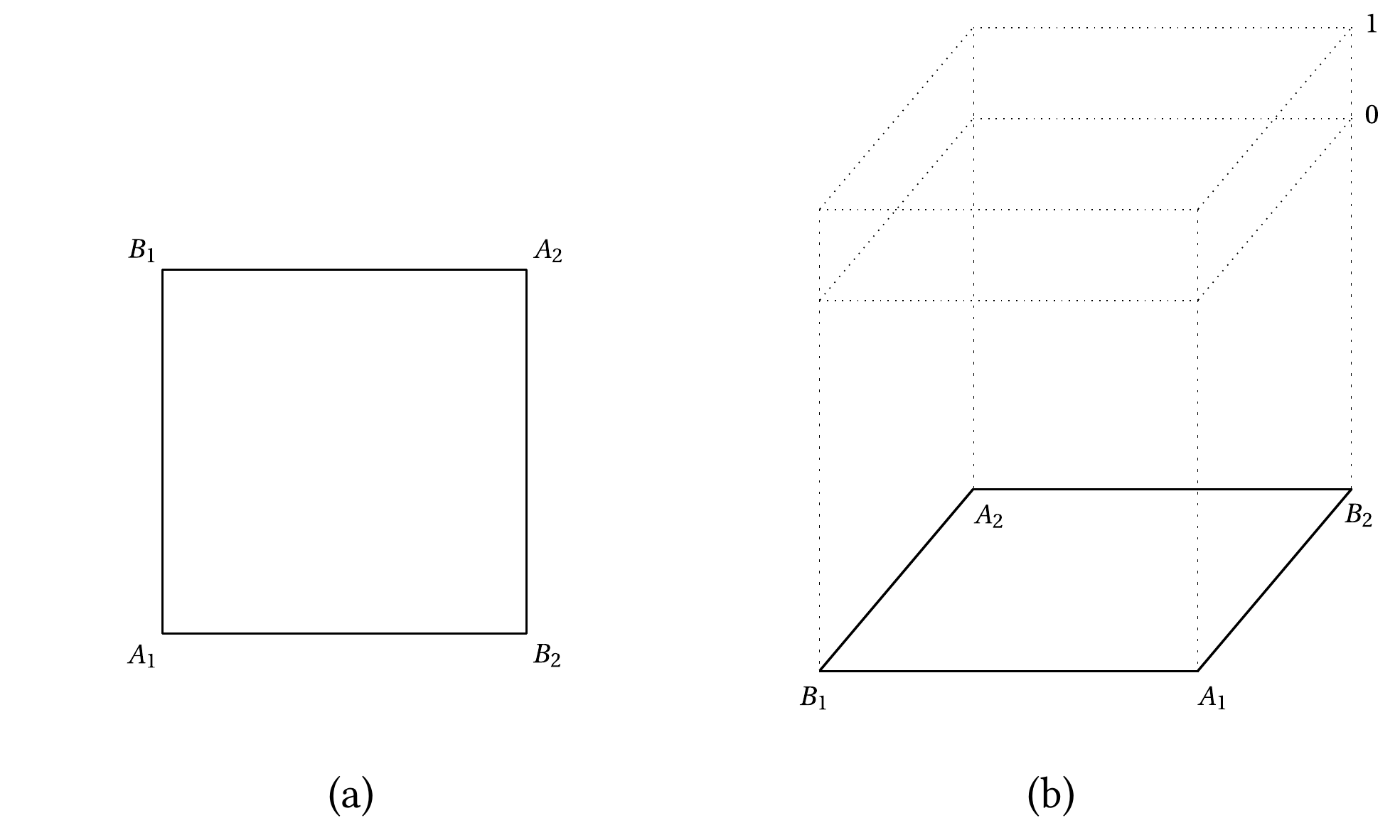}
\caption{Bundle diagram for CHSH scenario.}
\label{fig:3}
\end{figure}
\begin{figure}
\centering
\includegraphics[scale=0.35,angle=0]{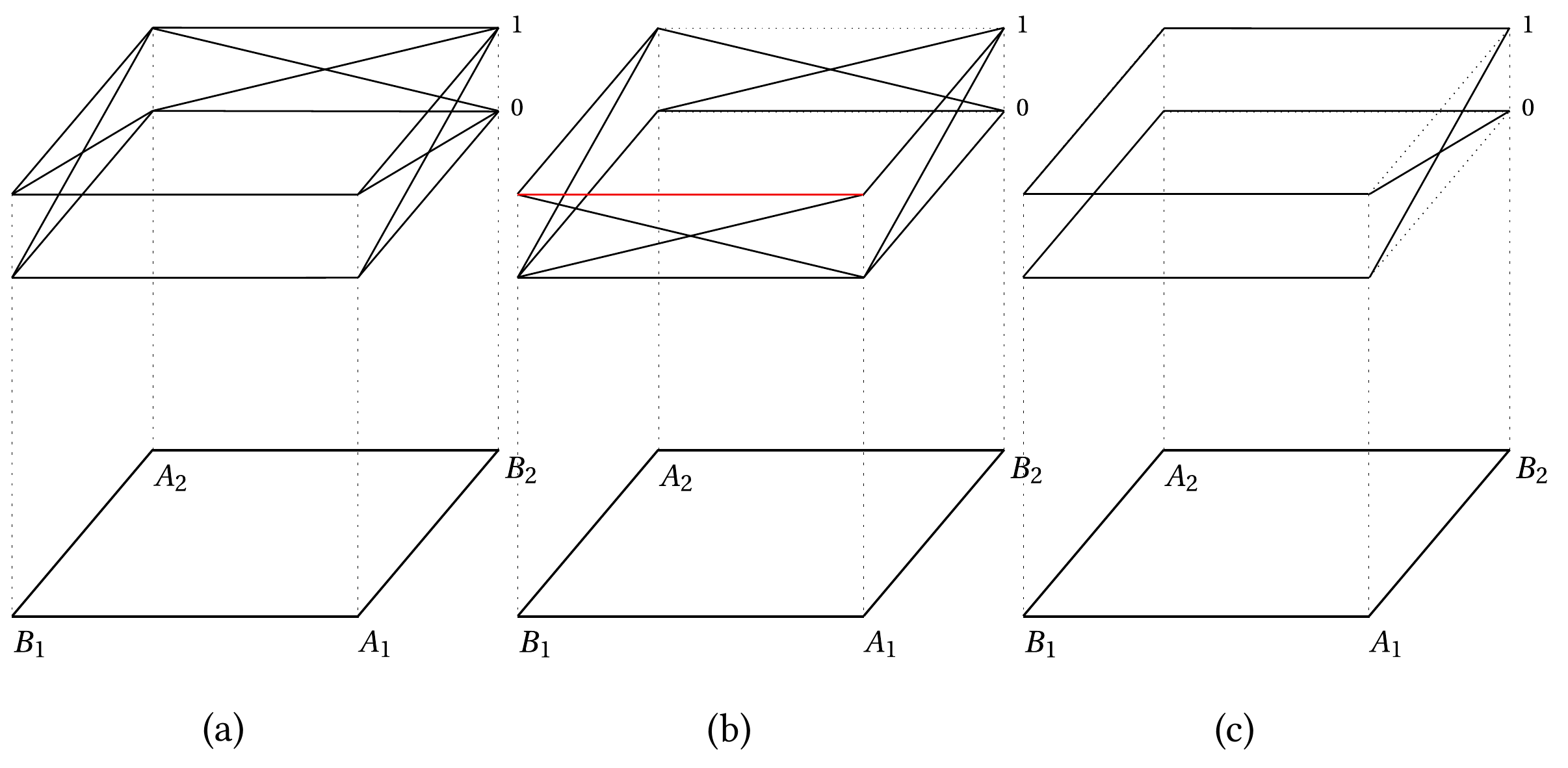}
\caption{Bundle diagram for the Bell's model (a); Hardy's model (b); and PR-Box (c).}
\label{fig:4}
\end{figure}

\section{Possibilistic Paradoxes and Logical Contextuality}\label{sec:IV}

As already mentioned, the violation of some NCI is a simple criterion to decide whether a given behavior is or not contextual. Despite being the most usual and also quite general, this is not the only possible strategy to demonstrate the occurrence of non-classicalities. In particular, in some specific cases it is possible to demonstrate quantum contextuality without the use of any inequality, that is, an inequality-free proof.

Formally, an inequality-free proof is based on the violation of logical constrains imposed by the assumption of NC by the possibilistic structure of the behavior. Similar to what occurs in proofs based on the violation of NCIs, in general only a subset of the set of possibilities is sufficient to obtain contradictions with NC instead of the whole behavior. Such subsets of possibilities define what we call \textit{possibilistic paradoxes}.

\begin{defn}\label{defn:IV.1}
A \textit{possibilistic paradox} (PP) is a set of possibilistic conditions whose occurrence implies contextuality.
\end{defn}

In a sense, the concept of PP is the possibilistic analogue of NCI. This is because the occurrence of a PP implies not only contextuality but also LC, since it provides a purely possibilistic contradiction with NC. In addition, any inequality-free proof is associated to at least one PP.

\begin{example}[Hardy PP]\label{example:IV.1} Hardy's inequality-free proof \cite{Hardy92,Hardy93} is based on the occurrence of the following PP:
\begin{eqnarray}\label{eq:16}
\begin{cases}
p_{11}(1,1)>0 \\
p_{12}(1,0)=0 \\
p_{21}(0,1)=0 \\
p_{22}(1,1)=0
\end{cases},
\end{eqnarray}
where $p_{\mu\nu}$ denotes the joint probability for the context $\{A_\mu,B_\nu\}$ in the Bell scenario $(2,2,2)$, with Alice's and Bob's measurement $\{A_1,A_2\}$ and $\{B_1,B_2\}$, respectively. To demonstrate that these conditions actually define a PP, just notice that the second, forth, and third conditions read respectively as: $A_1=1\implies B_2=1$, $B_2=1\implies A_2=0$, and $A_2=0\implies B_1=0$. Hence, for any global assignment compatible with these constrains we have $A_1=1\implies B_1=0$, meaning that $p_{11}(1,1)=0$, which contradicts the first condition. Therefore, the occurrence of (\ref{eq:16}) implies LC.
\end{example}

\begin{example}[Cabello \textit{et. al} PP] Cabello \textit{et al.} inequality-free proof \cite{CBTCB13} is based on the occurrence of the following PP:
\begin{eqnarray}\label{eq:17}
\begin{cases}
p_{1}(0,1)>0 \\
p_{2}(1,1)=0 \\
p_{3}(0,0)=0 \\
p_{4}(1,1)=0 \\
p_{5}(0,0)=0
\end{cases},
\end{eqnarray}
where $p_{\mu}$ denotes the joint probability for the context $\{M_\mu,M_{\mu+1}\}$ in the KCBS scenario. To demonstrate that (\ref{eq:17}) is actually a PP, just notice that the second, third, forth, and fifth conditions read respectively as: $M_2=1\implies M_3=0$, $M_3=0\implies M_4=1$, $M_4=1\implies M_5=0$, and $M_5=0 \implies M_1=1$. Hence, for any global assignment compatible with these constrains we have $M_2=1\implies M_1=1$, meaning that $p_1(0,1)=0$, which contradicts the first condition. Therefore, the occurrence of (\ref{eq:17}) implies LC.
\end{example}

\subsection{Necessary and Sufficient Conditions for LC}

Herein we aim at understanding, as general as possible, the connection between PP and LC. The building blocks of our discussion will be the $n$-cycle scenarios. As already mentioned, an $n$-cycle scenario consists of $n$ dichotomic measurements, $\{M_1,\dots,M_n\}$, $O=\{0,1\}$, whose contexts are $C_i:=\{M_i,M_{i+1}\}$, where the sums in indexes are taken modulo $n$. The following theorem states necessary and sufficient conditions for LC in such scenarios.

\begin{rem}
In what follows, whenever we refer to the $n$-cycle scenarios, the sums in the indexes will be taken modulo $n$, where $n+1=1$. A joint probability distribution for the context $C_i$ will be denoted by $p_i$.
\end{rem}

\begin{thm}\label{thm:2}
A behavior for the $n$-cycle scenario is LC if, and only if, there exist $(a,b)\in O^2$ and $(\alpha_1,\dots,\alpha_{n-2})\in O^{n-2}$ such that
\begin{eqnarray}\label{eq:18}
\begin{cases}
p_i(a,b)>0 \\
p_{i+1}(b,\alpha_1)=0 \\
p_{i+2}(\lnot \alpha_1,\alpha_2)=0 \\
\vdots \\
p_{i+n-1}(\lnot \alpha_{n-2},a)=0
\end{cases},
\end{eqnarray}
for some $1\leq i\leq n$.
\end{thm}
\begin{proof}
First of all, let us demonstrate that the conditions (\ref{eq:18}) defines a PP. The second, third, $\dots$ conditions can be respectively read as: $M_{i+1}=b \implies M_{i+2}=\lnot \alpha_1$, $M_{i+2}=\lnot\alpha_1\implies M_{i+3}=\lnot \alpha_2$, $M_{i+3}=\lnot\alpha_1\implies M_{i+4}=\lnot \alpha_2$, $\dots$, $M_{i+n-1}=\lnot\alpha_{n-2}\implies M_{i}=\lnot a$. Therefore, for any global assignment compatible with these constrains we have $M_{i+1}=b \implies M_{i}=\lnot a$, meaning that $p_i(a,b)=0$, which contradicts the first condition. Therefore, the occurrence of (\ref{eq:18}) implies LC.

To prove the converse, consider the contrapositive of the statement, \textit{i.e.} let us prove that non-occurrence of (\ref{eq:18}) implies logical noncontextuality. Non-occurrence of (\ref{eq:18}) means that
\begin{align}\label{eq:19}
\bar{p}_{i+1}(b,\alpha_1)&\lor\left(\bigvee_{j=2}^{n-2}\bar{p}_{i+j}(\lnot\alpha_{j-1},\alpha_j)\right)\nonumber\\ &\lor \bar{p}_{i+n-1}(\lnot\alpha_{n-2},a)=1,
\end{align}
for all $(\alpha_1,\dots,\alpha_{n-2})\in O^{n-2}$ whenever $\bar{p}_i(a,b)=1$. For simplicity (but without loss of generality), consider $i=1$. Since $\bar{p}_1(a,b)=1$, ND [Eq. (\ref{eq:11})] implies that
\begin{eqnarray}\label{eq:20}
\begin{cases}
\bar{p}_2(b,\alpha_1)\lor\bar{p}_2(b,\lnot \alpha_1)=1 \\
\bar{p}_n(\alpha_{n-2},a)\lor\bar{p}_n(\lnot \alpha_{n-2},a)=1 \\
\end{cases},
\end{eqnarray}
for all $\alpha_1,\alpha_{n-2}\in O$. From normalization condition to $\bar{p}_3$ and ND, it is always possible to find $(\beta_1,\dots,\beta_{n-2})\in O^{n-2}$ such that
\begin{eqnarray}\label{eq:21}
\bar{p}_3(\beta_1,\beta_2)=\dots=\bar{p}_{n-1}(\beta_{n-3},\beta_{n-2})=1.
\end{eqnarray}
Since $\bar{p}_3(\beta_1,\beta_2)=\bar{p}_{n-1}(\beta_{n-3},\beta_{n-2})=1$, from ND follows that 
\begin{subeqnarray}\label{eq:22}
&&\begin{cases}
\bar{p}_2(b,\beta_1)\lor\bar{p}_2(b,\lnot \beta_1)=1 \\
\bar{p}_2(b,\beta_1)\lor\bar{p}_2(\lnot b, \beta_1)=1 \\
\end{cases},\\
&&\begin{cases}
\bar{p}_n(\beta_{n-2},a)\lor\bar{p}_n(\lnot \beta_{n-2},a)=1 \\
\bar{p}_n(\beta_{n-2},a)\lor\bar{p}_n(\beta_{n-2},\lnot a)=1 \\
\end{cases}.
\end{subeqnarray}

In (\ref{eq:22}) we have a set of Boolean equations that need to be simultaneously satisfied. A trivial solution of Eqs. (\ref{eq:22}) is
\begin{eqnarray}\label{eq:23}
\bar{p}_2(b,\beta_1)=\bar{p}_n(\beta_{n-2},a)=1.
\end{eqnarray}
From (\ref{eq:21}), (\ref{eq:23}), and the fact that $\bar{p}_1(a,b)=1$, it follows that, in this considered case, $(a,b,\beta_1,\dots,\beta_{n-2})$ is a global assignment. Hence, from theorem \ref{thm:1}, the associated behavior is logically NC. Suppose that the most non-trivial case occurs (the other possible cases may be proved in a similar way to what we will consider), that is
\begin{align}\label{eq:24}
\bar{p}_2(b,\lnot \beta_1)&=\bar{p}_2(\lnot b, \beta_1)\nonumber \\&=\bar{p}_n(\lnot \beta_{n-2},a)=\bar{p}_n(\beta_{n-2},\lnot a)=1,
\end{align}
and
\begin{eqnarray}\label{eq:25}
\bar{p}_2(b,\beta_1)=\bar{p}_n(\beta_{n-2},a)=0.
\end{eqnarray}
Taking $\alpha_1=\beta_1$ and $\alpha_{n-2}=\beta_{n-2}$ in (\ref{eq:19}), the non-occurrence of (\ref{eq:18}) reads
\begin{align}\label{eq:26}
\bar{p}_3(\lnot\beta_1,\alpha_2)&\lor\left(\bigvee_{j=4}^{n-2}\bar{p}_j(\lnot\alpha_{j-2},\alpha_{j-1})\right)\nonumber \\&\lor\bar{p}_{n-1}(\lnot \alpha_{n-3},\beta_{n-2})=1.
\end{align}
Since from (\ref{eq:24}), $\bar{p}_2(b,\lnot\beta_1)=1$, the ND condition implies that
\begin{eqnarray}\label{eq:27}
\bar{p}_3(\lnot \beta_1,\beta_2)=\dots=\bar{p}_n(\beta_{n-3},\beta_{n-2}')=1.
\end{eqnarray}
If $\beta_{n-2}'=\lnot\beta_{n-2}$, then $(a,b,\lnot\beta_1,\beta_2,\dots,\beta_{n-1},\lnot \beta_{n-2})$ is a global assignment, and then, from theorem \ref{thm:1}, the associated behavior is logically NC. Consider the other possible case, \textit{i.e.} 
\begin{eqnarray}\label{eq:28}
\bar{p}_{n-1}(\beta_{n-3},\lnot\beta_{n-2})=0.
\end{eqnarray}
Since $\bar{p}_n(\lnot \beta_{n-2},a)=1$, from (\ref{eq:11}) and (\ref{eq:20}), we must have that
\begin{eqnarray}\label{eq:29}
\bar{p}_{n-1}(\lnot \beta_{n-3},\lnot \beta_{n-2})=1.
\end{eqnarray}
Applying the ND condition (\ref{eq:4}) to (\ref{eq:29}):
\begin{eqnarray}\label{eq:30}
\bar{p}_{n-2}(\beta_{n-4},\lnot\beta_{n-3})\lor \bar{p}_{n-2}(\lnot \beta_{n-4},\lnot\beta_{n-3})=1.
\end{eqnarray}
If $\bar{p}_{n-2}(\beta_{n-4},\lnot\beta_{n-3})=1$, then $(a,b,\lnot\beta_1,\beta_2,\dots,\beta_{n-4},\lnot\beta_{n-3},\lnot \beta_{n-2})$ is a global assignment, and, from theorem \ref{thm:1}, we have logical NC. Consider again the case where it does not occur, \textit{i.e.}
\begin{eqnarray}\label{eq:31}
\bar{p}_{n-2}(\beta_{n-4},\lnot\beta_{n-3})=0.
\end{eqnarray}
If we keep repeating this process, always choosing solutions that do not allow the construction of a global assignment, we will have
\begin{align}\label{eq:32}
\bar{p}_{n-1}(\beta_{n-3},\lnot\beta_{n-2})&=\bar{p}_{n-2}(\beta_{n-4},\lnot\beta_{n-3})\nonumber \\ &=\dots=\bar{p}_3(\beta_1,\lnot \beta_2)=0.
\end{align}
However, applying De-Morgan's law to (\ref{eq:32}), it is easy to see that it contradicts the initial hypothesis of non-occurrence of (\ref{eq:18}). Therefore, the non-occurrence of the conditions (\ref{eq:18}) in the $n$-cycle scenario implies logical NC.
\end{proof}

An immediate consequence of theorem \ref{thm:2} is that in the bipartite Bell scenario $(2,2,2)$, the occurrence of the generalization of the Hardy's PP (\ref{eq:16}) are necessary and sufficient for logical nonlocality.

\begin{rem}
In what follows, whenever we refer to a bipartite Bell scenario $(2,k,2)$, we will consider that one part (Alice) has the measurements $\{A_1,\dots,A_k\}$, while the other (Bob) has $\{B_1,\dots,B_k\}$. We will denote by $p_{\mu\nu}$ the joint probability distribution for the context $\{A_\mu,B_\nu\}$.
\end{rem}

\begin{cor}\label{cor:1}
For the Bell scenario $(2,2,2)$, a behavior is LC if, and only if, there exists a joint outcome $(a,b)\in O^2$ and $(\alpha_1,\alpha_2)\in O^2$ such that
\begin{eqnarray}\label{eq:33}
\begin{cases}
p_{ij}(a,b)>0\\
p_{kj}(\alpha_1,b)=0 \\
p_{i\ell}(a,\alpha_2)=0 \\
p_{k\ell}(\lnot\alpha_1,\lnot\alpha_2)=0 \\
\end{cases},
\end{eqnarray}
where $i\neq k$ and $j\neq \ell$.
\end{cor}
\begin{proof}
Just notice that the Bell scenario $(2,2,2)$ is the 4-cycle scenario if one defines $M_1=A_1$, $M_2=B_1$, $M_3=A_2$, and $M_4=B_2$. Hence, theorem \ref{thm:2} can be immediately applied.
\end{proof}

A less trivial consequence of theorem \ref{thm:2} is that the PP (\ref{eq:18}) has a "universality" on simple scenarios. At this point, it is important to recall the Vorob'yev theorem \cite{Vorobev63,Vorobyev67}, which states that acyclic scenarios are contextuality-free \cite{BMC16}. The following theorem states that, in non-acyclic simple scenarios, we just need to look for their cycles induced subgraphs in order to detect LC.

\begin{thm}\label{thm:3}
In a simple scenario which is not free of cycles, the occurrence of the PP (\ref{eq:18}) in one of its cycles is necessary and sufficient for LC.
\end{thm}
\begin{proof}
From the same argument presented in the proof of theorem \ref{thm:2}, the occurrence of (\ref{eq:18}) implies LC. To prove the converse, consider the contrapositive of the statement, \textit{i.e.} let us assume that no such PPs occur and then show that it implies logical NC.

First of all, consider the simplest situation, where the scenario has only one cycle. In this considered case, the non-occurrence of (\ref{eq:18}) implies that it is always possible to find $(a_1, \dots, a_n)\in O^n$ such that $\bar{p}_1(a_1,a_2)= \dots =\bar{p}_n(a_n, a_1)=1$, where $\bar{p}_k$ denotes the joint possibility distribution for the context $\{M_k,M_{k+1}\}$ in the cycle. We can always extend any of these assignments to the other measurements in the scenario by using the ND condition [Eq. (\ref{eq:11})]. Therefore, the non-occurrence of (\ref{eq:18}) in the cycle implies logical NC. The same can be concluded for a scenario with an arbitrary number of disjoint cycles. 

\begin{figure}
\centering
\includegraphics[scale=0.55,angle=0]{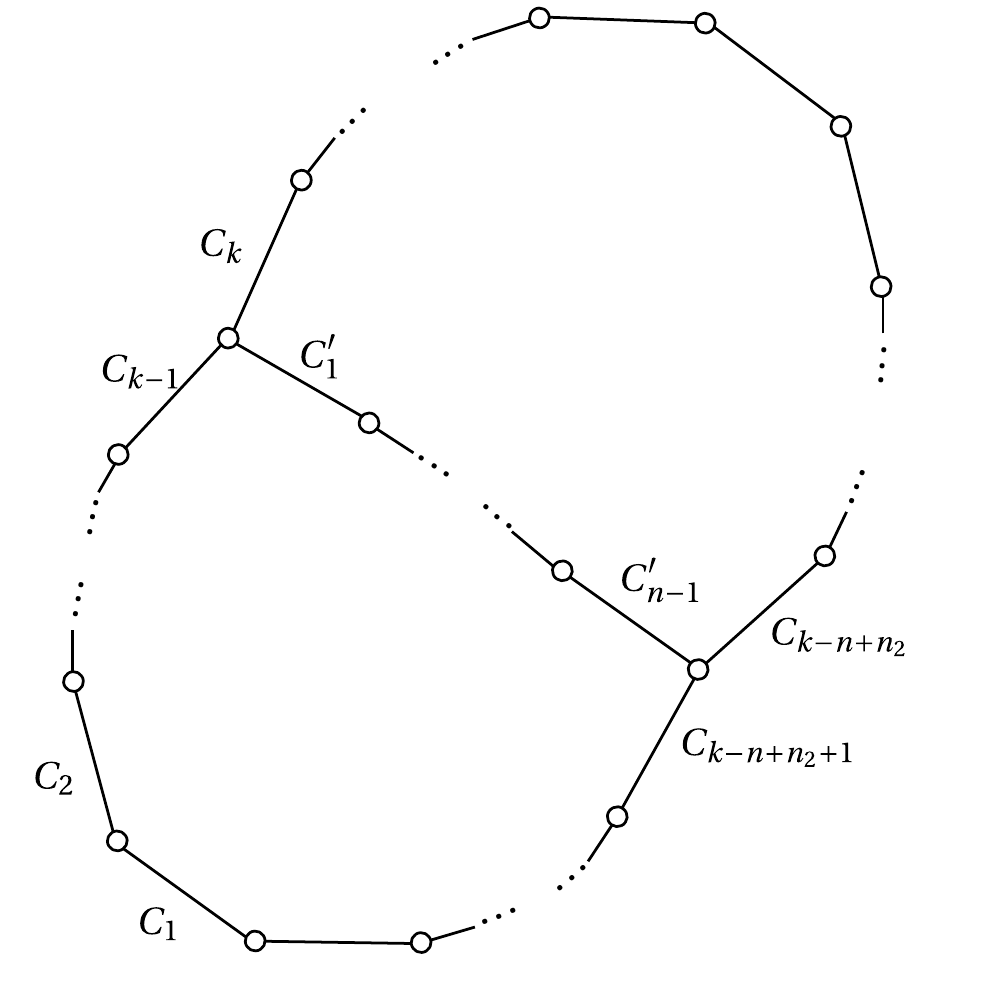}
\caption{Compatibility graph of two coupled cycles with $n_1$ and $n_2$ vertices sharing $n$ of those.}
\label{fig:5}
\end{figure}

Finally, consider the case where there are coupled cycles. For simplicity consider two cycles with $n_1$ and $n_2$ vertices sharing $n$ of those. In this scenario it is easy to see that there are three possible cycles: two cycles with $n_1$ and $n_2$ vertices and an external cycle with $N:=n_1+n_2+2(1-n)$ vertices (see Fig. \ref{fig:3}). From Theorem \ref{thm:2}, the absence of PP (\ref{eq:18}) in the external cycle implies that it is possible to give an assignment $(a_1,\dots,a_N)$ such that $\bar{p}_1(a_1,a_2)=\bar{p}_2(a_2,a_3)=\dots=\bar{p}_N(a_{N},a_1)= 1$, where $\bar{p}_i$ denotes the joint possibility distribution for the context $\{M_i,M_{i+1}\}$ in the external cycle. Applying the non-disturbance condition, we can extend this assignment in the following way: $\bar{p}'_1(a_k,\alpha_2)=\bar{p}'_2 (\alpha_2,\alpha_3)=\dots=\bar{p}'_{n-2}(\alpha_{n-2},\alpha_{n-1})=\bar{p}'_{n-1}(\alpha_{n-1},\alpha_n)$, where $\bar{p}'_i$ denotes the joint probability distributions for the context $C_i'$ (see Fig. \ref{fig:3}). If $\alpha_n=a_{k+n_2-n+1}$, then $(a_1,\dots,a_N,\alpha_2,\dots,\alpha_{n-1})$ is a global assignment, and the theorem is proved. Let us assume that the non-trivial case occurs, that is $p'_{n-1}(\alpha_{n-1},a_{k+n_2-n+1})=0$. Since $\bar{p}_{k+n_2-n+1}(a_{k+n_2+n},a_{k+n_2+n+1})=1$, from ND condition, $\bar{p}'_{n-1}(\lnot \alpha_{n-1},a_{k+n_2+n+1})=1$. Applying ND condition to $\bar{p}'_{n-1}(\lnot \alpha_{n-1},a_{k+n_2+n+1})=1$ we have
\begin{eqnarray}\label{eq:34}
\bar{p}'_{n-2}(\alpha_{n-2},\lnot \alpha_{n-1})\lor \bar{p}'_{n-2}(\lnot\alpha_{n-2},\lnot \alpha_{n-1})=1.
\end{eqnarray}
If $\bar{p}'_{n-2}(\alpha_{n-2},\lnot \alpha_{n-1})=1$, $(a_1,\dots,a_n,\alpha_2,\dots,\alpha_{n-2},\lnot\alpha_{n-1})$ is a global assignment, and then the theorem is proved. Consider again the non-trivial case: $\bar{p}'_{n-2}(\alpha_{n-2}.\lnot \alpha_{n-1})=0$. Until here we have
\begin{eqnarray}\label{eq:35}
\bar{p}'_{n-1}(\alpha_{n-1},a_{k+n_2-n+1})=\bar{p}'_{n-2}(\alpha_{n-2},\lnot\alpha_1)=0.
\end{eqnarray}
Hence, it is easy to see that if we continue this process, always taking the solution in which the theorem is not proved, at the end of the day we will construct a condition (\ref{eq:18}) in at least one of the cycles. Therefore, the absence of a PP (\ref{eq:18}) in all cycles implies logical NC.

This result can be extended to any simple scenario with an arbitrary number of cycles (coupled or not). For that, it is enough that we apply the same arguments above for each pair of coupled cycles of the scenario.
\end{proof}

\begin{cor}\label{cor:2}
For the bipartite Bell scenario $(2,k,2)$, a behavior is LC if, and only if, the following conditions,
\begin{eqnarray}\label{eq:36}
\begin{cases}
p_{ij}(a,b)>0\\
p_{mj}(\alpha_1,b)=0 \\
p_{i\ell}(a,\alpha_2)=0 \\
p_{m\ell}(\lnot\alpha_1,\lnot\alpha_2)=0 \\
\end{cases},
\end{eqnarray}
hold, where $i\neq m$ and $j\neq \ell$.
\end{cor}
\begin{proof}
The Bell scenario $(2,k,2)$ consists of $[k(k-1)/2]^2$ coupled cycles with four vertices (see Fig. \ref{fig:6}). Notice that for any pair of coupled 4-cycles, the external cycle is also a 4-cycle (see proof of theorem \ref{thm:3} and Fig. \ref{fig:6}). Each choice of indexes in (\ref{eq:36}) defines a possible PP (\ref{eq:18}) in one of the possible 4-cycle subgraphs. Therefore the assumption follows as a directly consequence of theorem \ref{thm:3} and corollary \ref{cor:1}.
\end{proof}

\begin{figure}
\centering
\includegraphics[scale=0.31]{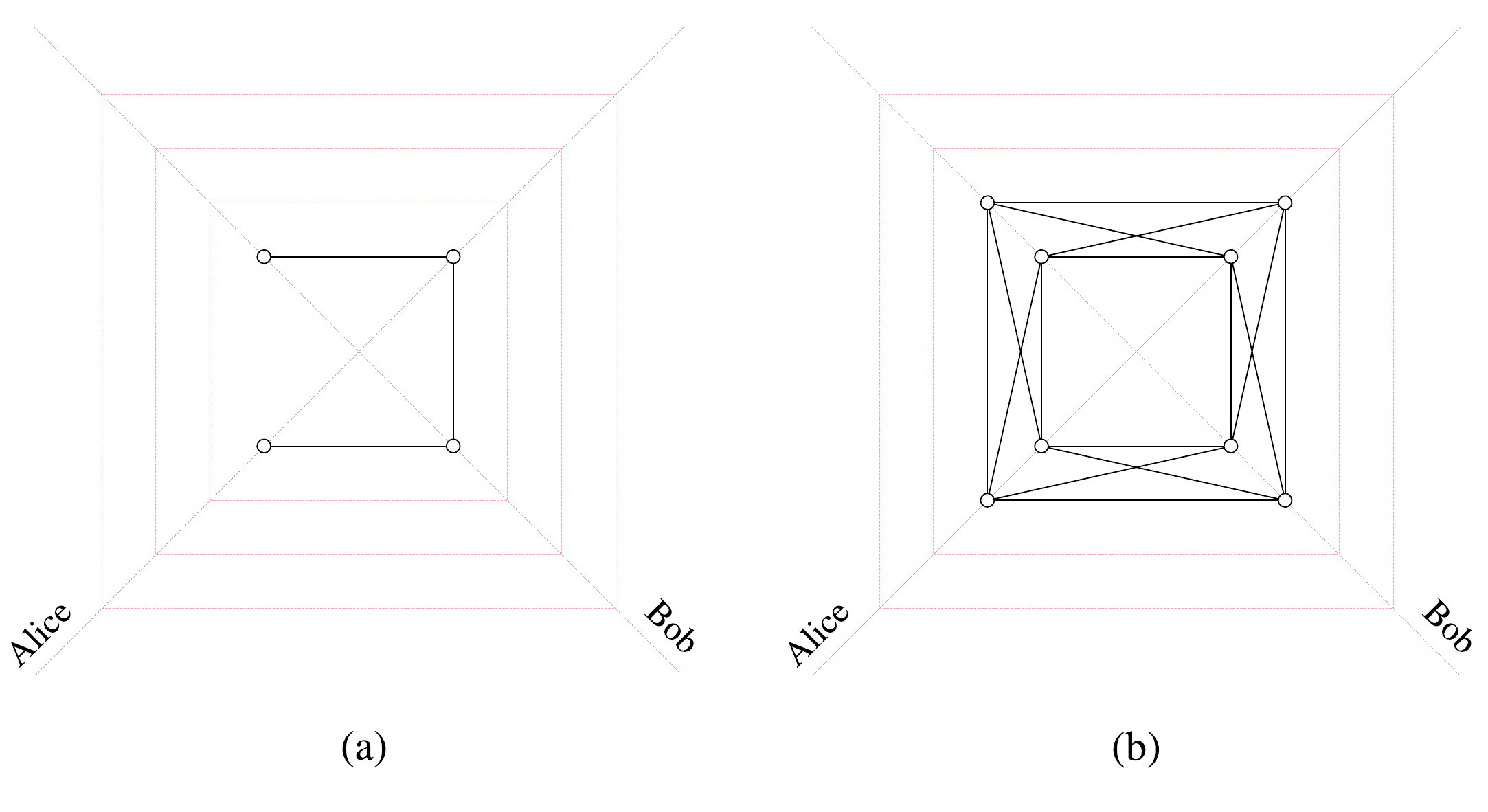}
\caption{Compatibility graph of bipartite Bell scenarios. (a) $(2,2,\ell)$; (b) $(2,4,\ell)$.}
\label{fig:6}
\end{figure}

An useful way to draw the compatibility graph of bipartite Bell scenarios emphasizing their cycles is to represent the measurements by vertices on two different diagonal lines. Since this graph is bipartite, we connect two vertices if, and only if, they belong to different lines (see Fig. \ref{fig:6}). From that construction, it is very clear that a Bell-type scenario $(2,k,2)$ has cycles of up to $2k$ vertices as induced subgraphs. In each of these cycles, it is possible to construct a PP (\ref{eq:18}). There is an inequality-free proof for quantum nonlocality based on this construction, called \textit{Hardy ladder proof}, since the Boschi \textit{et al.} contribution \cite{BBMH97}. It should be stressed that the occurrence of these PP implies the occurrence of (\ref{eq:36}). This fact is a consequence of the completeness given by Corolary \ref{cor:2}, or can be proved directly (\textit{cf.} Ref. \cite{MF12}).

\subsection{Several Outcomes Possibilistic Paradoxes}

So far we have demonstrated that the occurrence of a specific kind of PPs is a necessary and sufficient condition for LC. Such a result applies for every simple scenario with dichotomic measurements. Our goal in this paper is to state results as general as possible. So, a natural question is whether it is possible to extend these results for scenarios where the measurements are not dichotomic. As argued by Mansfield and Fritz \cite{MF12} to the Bell scenarios $(2,k,\ell)$ the answer is no. However, we can shed some light on this issue, at least in a restricted class of scenarios. Such scenarios are the "generalized" $n$-cycle ones, where we have measurements with more than two possible outcomes.

\begin{thm}\label{thm:4}
In a generalized $n$-cycle scenario, where each measurement has $\ell$ possible outcomes, a behavior is LC if, and only if there exists $(a,b)\in O^2$ such that
\begin{eqnarray}\label{eq:37}
\bar{p}_{\mu}(a,b)=1,
\end{eqnarray}
together with $(\alpha_i^1,\dots,\alpha_i^\ell)\in O^\ell$ such that
\begin{align}\label{eq:38}
\bigvee_{i=1}^{m_1} \bar{p}_{\mu+1}(b,\alpha_1^i) &\lor \bigvee_{i=m_1+1}^\ell\bigvee_{j=1}^{m_2} \bar{p}_{\mu+2}(\alpha_1^i,\alpha_2^j) \nonumber \\ &\lor \dots \lor \bigvee_{i=m_\nu+1}^\ell \bigvee_{j=1}^{m_{\nu+1}} \bar{p}_{\mu+\nu-1}(\alpha_\nu^i,\alpha_{\nu+1}^j)\nonumber \\ &\lor\dots\lor \bigvee_{i=m_{n-2}+1}^{\ell} \bar{p}_{\mu+n-1}(\alpha_{n-2}^i,a)=0,
\end{align}
$\alpha_i^x\neq \alpha_i^y$ if $x\neq y$, $m_i\leq \ell$, $i\leq n-2$.
\end{thm}
\begin{proof}
From arguments similar to those presented so far, it is straightforward to verify that the simultaneous occurrence of conditions in Eqs. (\ref{eq:37}) and (\ref{eq:38}) implies LC. More precisely, for any global assignment compatible with the condition of Eq. (\ref{eq:38}) we must have that $\bar{p}_\mu(a,b)=0$, which contradicts (\ref{eq:37}). Now suppose that a given behavior is LC. From theorem \ref{thm:2} it implies that some PP (\ref{eq:18}) must occurs, otherwise it would be possible to build a global assignment by using a pair of outcomes. However, PP (\ref{eq:18}) is equivalent to the conditions of Eq. (\ref{eq:37}) together with Eq. (\ref{eq:38}) taking $m_1=\dots=m_{n-2}=1$.
\end{proof}

If we restrict our attention to the bipartite Bell scenarios, we recover Mansfield and Fritz's results \cite{MF12}. In that paper, the authors demonstrated that the non-occurrence of a "coarse-grained Hardy paradox" is equivalent to LC in the Bell scenarios  $(2,2,\ell)$. This is exactly the same content of the theorem \ref{thm:4} if one takes $n=4$.

Theorem \ref{thm:4} states that the occurrence of a PP defined by (\ref{eq:37}) together with (\ref{eq:38}) is necessary and sufficient to LC in the generalized $n$-cycle scenarios. However, a much simpler way of constructing an inequality-free proof for contextuality, which is particularly interesting when dealing with quantum implementations, was proposed by Chen \textit{et al.} \cite{CCXSWK13}.

\begin{example}[Chen \textit{et al.} PP]
Chen \textit{et al.} inequality-free proof \cite{CCXSWK13} is constructed in a generalized $4$-cycle scenario with $\ell$ possible outcomes, $O=\{0,1,\dots,\ell-1\}$, and it is based on the following PP:
\begin{eqnarray}\label{eq:Chen}
\begin{cases}
p_i(a_i<a_{i+1})>0 \\
p_{i+1}(a_{i+1}>a_{i+2})=0 \\
p_{i+2}(a_{i+2}>a_{i+3})=0 \\
p_{i+3}(a_{i+3}>a_i)=0
\end{cases},
\end{eqnarray}
where
\begin{subeqnarray}
p_\mu(a_\mu< a_{\mu+1}):=\sum_{x<y}p_\mu(x,y),\\
p_\mu(a_\mu>a_{\mu+1}):=\sum_{x>y}p_\mu(x,y).
\end{subeqnarray}
In order to demonstrate that Eq. (\ref{eq:Chen}) actually defines a PP, consider, without loss of generality, $i=1$. Notice that if such conditions occurs, the last $n-1$ conditions imply that for any global assignment $(a_1,a_2,a_3,a_4)$ we must have that $a_2\leq a_3\leq a_4\leq a_1$. Therefore, we are enforced to conclude that $p_1(a_1<a_2)=0$, which contradicts the first condition. Therefore, the occurrence of the conditions of Eq. (\ref{eq:Chen}) implies LC.
\end{example}

\subsection{Strong Contextuality and PR-Boxes}

In Ref. \cite{Mansfield17} the author discusses some consequences and applications of the completeness of "Hardy-type PP" (see Example \ref{example:IV.1}) for logical nonlocality. One of those is that the Popescu-Rohrlich boxes [Example \ref{example:III.3}] are the only strongly nonlocal behaviors for the CHSH scenario. The following theorem generalizes this result for the $n$-cycle scenarios.

\begin{figure}
\centering
\includegraphics[scale=0.3]{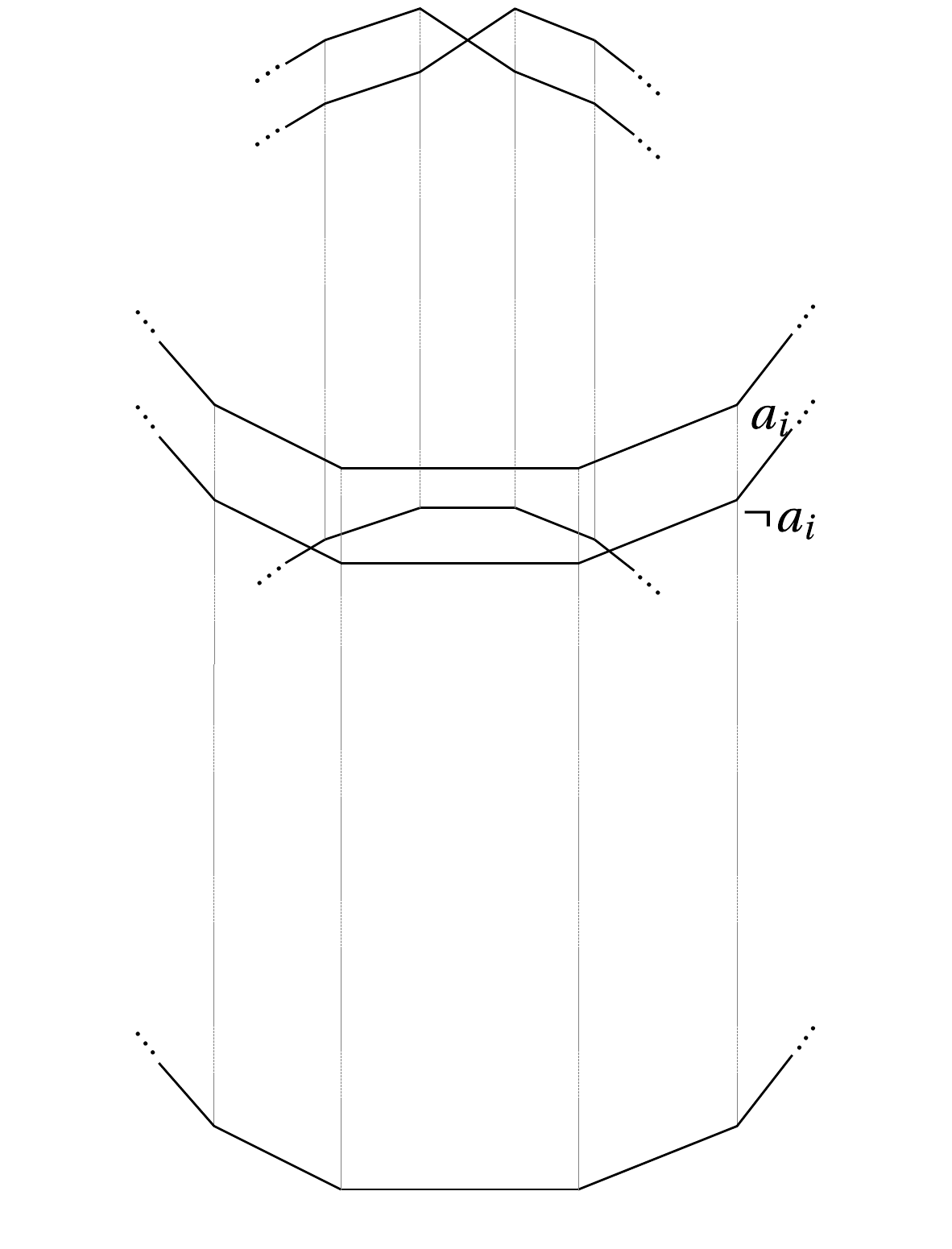}
\caption{Bundle diagram of a strong contextual behavior for an $n$-cycle.}
\label{fig:7}
\end{figure}

\begin{thm}\label{thm:5}
For an $n$-cycle scenario, every strongly contextual behavior is such that, for some $1\leq k\leq n$,
\begin{subeqnarray}\label{eq:41}
\bar{p}_{k}(a_k,\lnot a_{k+1})=\bar{p}_k(\lnot a_k, a_{k+1})=1, \\
\bar{p}_{k}(a_k,a_{k+1})=\bar{p}_k(\lnot a_k,\lnot a_{k+1})=0,
\end{subeqnarray}
and for $i\neq k$,
\begin{subeqnarray}\label{eq:42}
\bar{p}_{i}(a_i,a_{i+1})=\bar{p}_i(\lnot a_i,\lnot a_{i+1})=1, \\
\bar{p}_{i}(a_i,\lnot a_{i+1})=\bar{p}_i(\lnot a_i,a_{i+1})=0.
\end{subeqnarray}
\end{thm}
\begin{proof}
Recall that SC means that no global assignment can be defined [Def. \ref{defn:III.5}]. If a behavior for the $n$-cycle scenario is strongly contextual, then, from theorem \ref{thm:2}, every joint outcome with non-zero probability must be associated to a PP (\ref{eq:18}). Hence, for an arbitrary joint outcome $(a,b)$, $\bar{p}_1(a,b)=1$ implies that
\begin{eqnarray}\label{eq:43}
\bar{p}_2(b,\alpha_1)= \bar{p}_3(\lnot\alpha_{1},\alpha_{2})=\dots=\bar{p}_{n}(\lnot \alpha_{n-2},b)=0.
\end{eqnarray}
for some $(\alpha_1,\dots,\alpha_{n-2})\in O^{n-2}$. Since $\bar{p}_1(a,b)=1$, from (\ref{eq:44}) and ND condition [Eq. (\ref{eq:11})] follows
\begin{eqnarray}\label{eq:44}
\bar{p}_2(b,\lnot\alpha_1)&=&\bar{p}_3(\lnot \alpha_2,\lnot\alpha_3) \nonumber \\ &=&\dots=\bar{p}_{n-1}(\lnot\alpha_{n-2},\lnot a)=1,
\end{eqnarray}
and 
\begin{eqnarray}\label{eq:45}
\bar{p}_2(\lnot b,\alpha_1)&=&\bar{p}_3(\alpha_2,\alpha_3) \nonumber \\ &=&\dots=\bar{p}_{n-1}(\alpha_{n-2}, a)=1.
\end{eqnarray}
Since we suppose SC, we are enforced to impose that
\begin{eqnarray}\label{eq:46}
\bar{p}_1(\lnot a,b)=\bar{p}_1(a,\lnot b)=0,
\end{eqnarray}
otherwise $(a,\lnot b,\alpha_1,\dots,\alpha_{n-2})$ or $(\lnot a,b,\alpha_1,\dots,\alpha_{n-2})$ would be global assignments. Therefore, from (\ref{eq:44}--\ref{eq:46}) and ND condition, we must have that
\begin{eqnarray}\label{eq:47}
\bar{p}_1(\lnot a,\lnot b)=1.
\end{eqnarray}
Since the joint outcome $(\lnot a,\lnot b)$ must also be associated to a PP (\ref{eq:18}), from (\ref{eq:44}--\ref{eq:46}), we must have that
\begin{eqnarray}\label{eq:48}
\bar{p}_2(\lnot b,\lnot\alpha_1)&=&\bar{p}_3(\alpha_1,\lnot\alpha_2)\nonumber \\&=&\dots=\bar{p}_{n}(\alpha_{n-2},\lnot a)=0.
\end{eqnarray}
Defining $a_1=a$, $a_2=b$, $a_3=\lnot \alpha_1$, $a_4=\lnot \alpha_2$, and so on, the conditions (\ref{eq:43}--\ref{eq:48}) is exactly the same of (\ref{eq:41}--\ref{eq:42}) with $k=n$.
\end{proof}

It is a well-known fact that the Popescu-Rorhlich box is not compatible with quantum predictions \cite{Scarani19}. In fact, this behavior maximally violates the noncontextuality inequality for the CHSH scenario, thus violating its Tsirelson's bound \cite{Tsirelson93}. The same can be concluded for the $n$-cycle scenarios, where the Tsirelson's bound is given by Eq. (\ref{eq:7}). From theorem \ref{thm:5}, in any SC behavior the conditions (\ref{eq:41}) and (\ref{eq:42}) occur. In addition, it is easy to verify that their occurrence implies maximal violations of the $n$-cycle noncontextuality inequalities [Eq. (\ref{eq:6})], and consequently the Tsirelson's bound. Therefore, we have demonstrated the following corollary.

\begin{cor}
There is no strongly contextual quantum behavior for the $n$-cycle scenarios.
\end{cor}

In addition to the fact that they are "post-quantum", the SC behaviors of the $n$-cycle scenarios also resemble the PR-boxes [Example \ref{example:III.3}] in the fact that their bundle diagrammatic representations are "Möbius strips" (see Fig. \ref{fig:7}), exactly as happens in the CHSH scenario (see Fig. \ref{fig:4} (c)).

\subsection{Inequality-Free proofs in the $n$-cycle Scenarios}

In theorem \ref{thm:2} we demonstrated that the occurrence of the PP (\ref{eq:18}) is a necessary and sufficient condition for a given ND behavior in the $n$-cycle scenarios exhibit LC. Herein, we want to demonstrate how such PPs can be realized by quantum mechanical systems.

\begin{thm}\label{thm:6}
For the $n$-cycle scenarios, $n>3$, it is always possible to construct quantum behaviors where a PP (\ref{eq:18}) occurs. For $n$ odd the simplest proof uses a qutrit system ($\mathbb{C}^3$), and works with any pure state. For $n$ even, in turn, the simplest proof uses an 2-qubit system ($\mathbb{C}^2\otimes\mathbb{C}^2$), and works with any pure state which is neither a product nor a maximally entangled state.
\end{thm}
\begin{proof}
Let $n>3$ odd. Consider a qutrit system prepared in an arbitrary pure state $\ket{\eta}$. Define $\{\ket{v_1},\dots,\ket{v_n}\}\subset\mathbb{C}^3$ such that $\ket{v_3}$ is arbitrary since that $\bracket{v_3}{\eta}\neq 0$, and
\begin{subeqnarray}\label{eq:49}
\slabel{sub1}\ket{v_4}&=&\frac{\mathbb{I}-\ket{v_3}\bra{v_3}}{\sqrt{1-|\bracket{\eta}{v_3}|^2}}\ket{\eta}, \\ \slabel{sub2}
\ket{v_k}&=&\mathrm{R}(\theta_k,\ket{v_{k-1}})\ket{v_{k-2}}, \\ \slabel{sub3}
\ket{v_{k+1}}&=&\frac{\mathbb{I}-\ket{v_k}\bra{v_k}}{\sqrt{1-|\bracket{\eta}{v_{k}}|^2}}\ket{\eta},\\
\slabel{sub4}\ket{v_1}&=&\frac{\mathbb{I}-\ket{v_n}\bra{v_n}}{\sqrt{1-|\bracket{\eta}{v_n}|^2}}\ket{\eta}, \\ \slabel{sub5}
\ket{v_2}&=&\frac{\ket{v_1}\times \ket{v_3}}{||\hspace{0.05cm}\ket{v_1}\times \ket{v_3}\hspace{0.05cm}||},
\end{subeqnarray}
for $k\geq 5$ odd, $\theta_k\notin\{\ell\pi:\ell\in\mathbb{Z}\}$. $\mathrm{R}(\theta,\ket{x})$ denotes the rotation matrix by an angle $\theta$ around $\ket{x}$, and $\times$ denotes the usual three-dimensional cross product. From the above construction, one can verify that $\bracket{v_{k}}{v_{k+1}}=\bracket{v_1}{v_n}=0$. Hence, the compatibility relations for the odd-cycle scenario are well-defined if one consider measurements defined by projectors on vectors $\ket{v_i}$, that is $M_i:=\ket{v_i}\bra{v_i}$. From (\ref{eq:49}) it follows that
\begin{align}\label{eq:50}
\begin{cases}
p_1(0,1)\geq 0 \\
p_2(1,1)=0 \\
p_3(0,0)=0 \\
\vdots \\
p_{2k}(1,1)=0 \\
p_{2k+1}(0,0)=0 \\
\vdots \\ 
p_n(0,0)=0
\end{cases}.
\end{align}
The conditions above define a PP whenever $p_1(0,1)\neq 0$. With a suitable choice of angles, it is always possible to satisfy such a condition.

Now let $n\geq 4$ even. In this considered case, an 4-dimensional Hilbert space is needed \cite{KRK12}. For simplicity, consider a 2-qubit system in a arbitrary pure state $\ket{\eta}$. The Schmidt decomposition of $\ket{\eta}$ can be written as
\begin{eqnarray}\label{eq:51}
\ket{\eta}=\cos\alpha\ket{e_1}\otimes\ket{f_1}+\sin\alpha\ket{e_2}\otimes\ket{f_2},
\end{eqnarray}
where $\{\ket{e_1},\ket{e_2}\}$ and $\{\ket{f_1},\ket{f_2}\}$ are orthonormal basis of $\mathbb{C}^2$, and $\alpha \in [0,\pi/4]$. For $1\leq k\leq n/2$, define 
\begin{subeqnarray}\label{eq:52}
\ket{p_k}=\frac{(-1)^k(\cos\alpha)^{k-\frac{1}{2}}\ket{e_1}+\ii(\sin\alpha)^{k-\frac{1}{2}}\ket{e_2}}{\sqrt{(\cos\alpha)^{2k-1}+(\sin\alpha)^{2k-1}}}, \\ \ket{q_k}=\frac{(-1)^k(\cos\alpha)^{k-\frac{1}{2}}\ket{f_1}+\ii(\sin\alpha)^{k-\frac{1}{2}}\ket{f_2}}{\sqrt{(\cos\alpha)^{2k-1}+(\sin\alpha)^{2k-1}}},
\end{subeqnarray}
and the projective measurements
\begin{subeqnarray}\label{eq:53}
M_{\frac{n}{2}-k}&=&\ket{p_{k+2}}\bra{p_{k+2}}\otimes\mathbb{I}_2,\\
M_{\frac{n}{2}-k+1}&=&\mathbb{I}_2\otimes\ket{q_{k+1}}\bra{q_{k+1}}, \\
M_{\frac{n}{2}+1}&=&\ket{p_1}\bra{p_1}\otimes \mathbb{I}_2,\\
M_{\frac{n}{2}+2}&=&\mathbb{I}_2\otimes\ket{q_1}\bra{q_1}, \\
M_{\frac{n}{2}+k+2}&=&\ket{p_{k+1}}\bra{p_{k+1}}\otimes\mathbb{I}_2,\\
M_{\frac{n}{2}+k+3}&=&\mathbb{I}_2\otimes\ket{q_{k+2}}\bra{q_{k+2}},
\end{subeqnarray}
where $1\leq k \leq (n-4)/2$ and $n+1:=1$. From the construction above, it straightforward to notice that the compatibility relations of the $n$-cycle scenario are well defined. In addition, one may verify that
\begin{eqnarray}\label{eq:54}
\begin{cases}
p_1(1,1)\geq 0 \\
p_2(1,0)=0 \\
\vdots \\
p_{\frac{n}{2}}(1,0)=0 \\
p_{\frac{n}{2}+1}(1,1)=0 \\
p_{\frac{n}{2}+2}(0,1)=0 \\
\vdots\\p_{n-1}(0,1)=0 \\ 
p_n(0,1)=0
\end{cases}.
\end{eqnarray}
The conditions above define a PP (\ref{eq:18}) whenever $p_1(1,1)\neq 0$. From (\ref{eq:51}--\ref{eq:53}),
\begin{eqnarray}\label{eq:55}
p_1(1,1)=\left(\frac{\cos\alpha(\sin\alpha)^{n-1}-\sin\alpha(\cos\alpha)^{n-1}}{(\cos\alpha)^{n-1}+(\sin\alpha)^{n-1}}\right)^2.
\end{eqnarray}
So, $p_1(1,1)\neq 0$ if, and only if, $\alpha\notin\{0,\pi/4\}$, that is $\ket{\eta}$ is neither a product or a maximally entangled state. The largest values of $p_1(1,1)$ in (\ref{eq:55}) are plotted in Fig. \ref{fig:8}.
\end{proof}

\begin{figure}
\centering
\includegraphics[scale=0.3]{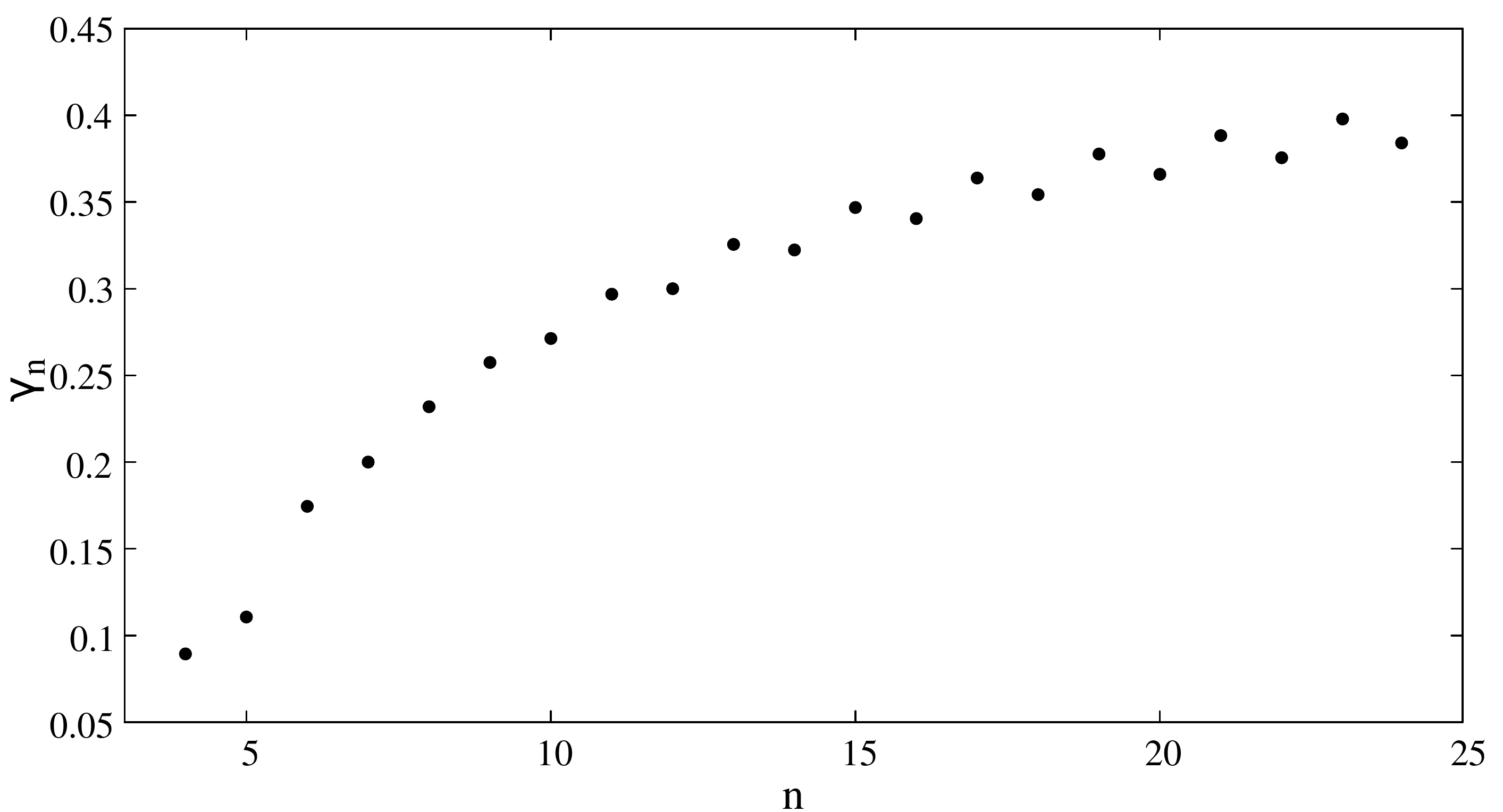}
\caption{Largest value of the non-zero probability, denoted by $\gamma_n$, in Eqs. (\ref{eq:50}) and (\ref{eq:54}).}
\label{fig:8}
\end{figure}

In the theorem above, the contradiction with NC occurs due to the fact that the probabilities $p_1(0,1)$ and $p_1(1,1)$, in Eqs. (\ref{eq:50}) and (\ref{eq:54}) respectively, can be strictly greater than zero. For $n$ odd, it is quite difficult to compute the largest possible value of $p_1(0,1)$ (which we will denote by $\gamma_n$). In fact, only for simple cases it is possible to compute it analytically. For instance, for $n=5$ we have \cite{CBTCB13}
\begin{equation}\label{eq:56}
\gamma_5=\max_{\alpha,\beta\in[0,\pi)} \frac{[\sin(2\alpha) \sin(2\beta)]^2}{[\cos\alpha\sin(2\beta)]^2+(2\sin\alpha)^2}=\frac{1}{9}.
\end{equation}
For $n=7$ the calculations are much more laborious, and one may verify that $\gamma_7=1/5$. For higher values of $n$ we numerically estimate $\gamma_n$ using Mathematica ® (see Fig. \ref{fig:8}). We verify that $\gamma_9\approx 0.257371$, which is greater than $(1+16/\sqrt{27})^{-1}$, claimed by Cabello \textit{et al.} \cite{CBTCB13} to be the value of $\gamma_9$. For the even-cycle scenarios, on the other hand, the choice of projectors is much less general. In fact, for $n$ odd we have $(n-1)/2$ "free parameters" ($(n-3)/2$ angles $\theta_k$, and the angle between $\ket{\eta}$ and $\ket{v_3}$), while for $n$ even we have only one (the angle $\alpha$). The corresponding $\gamma_n$ value is obtained by maximizing $p_1(1,1)$ in (\ref{eq:55}) (see Fig. \ref{fig:8}).

\section{Discussion and Further Steps}\label{sec:V}

In this paper we have investigated logical contextuality and inequality-free proofs. We demonstrated that the occurrence of specific sets of possibilistic conditions is a necessary and sufficient condition for logical contextuality in the $n$-cycle scenarios [Theorem \ref{thm:2}], for simple scenarios with dichotomic measurements [Theorem \ref{thm:3}], and for general $n$-cycle scenarios [Theorem \ref{thm:4}]. As a consequence of this result we concluded that the only strongly contextual behaviors for the $n$-cycle scenarios are generalizations of the PR-boxes [Theorem \ref{thm:5}]. Also for the $n$-cycle scenarios, we demonstrated that it is always possible to construct a "Hardy-like proof" to quantum contextuality with a qutrit system (for $n$ odd) and an 2-qubit system (for $n$ even).

The results that we have established in the present paper help to better understand the connection between logical contextuality and inequality-free proofs. In particular, we define the concept of possibilistic paradoxes [Def.  \ref{defn:IV.1}], that are the possibilistic analogue of noncontextuality inequalities. In this way, we find sets of possibilistic paradoxes whose occurrence is a necessary and sufficient condition for logical contextuality in the mentioned scenarios, which can be understood as the possible analogue of H-representations of the noncontextual polytope.

Since logical contextuality is the most general form of possibilistic non-classicality, our results can be used to explain any inequality-free proof to quantum contextuality or nonlocality. For instance, Hardy's nonlocality paradox \cite{Hardy92,Hardy93} can be recovered by theorem \ref{thm:6} by taking $n=4$. In the same way, Cabello \textit{et al.} proof can be recovered by taking $n=5$, $\ket{\eta}=(1,1,1)^T/\sqrt{3}$, $\ket{v_3}=(1,0,0)^T$,and $\theta_5=\pi/4$ in (\ref{eq:49}), where $T$ means transposition. Hardy ladder proof \cite{BBMH97}, in turn, can be constructed by using the measurement operators we have defined in the proof of theorem \ref{thm:6} for $n$ even. It is not yet clear, however, how possibilistic paradoxes can be realized by quantum theory in more general scenarios with non-dichotomic measurements.

The results presented in this work do not apply to scenarios whose contexts have more than two measurements. We believe that results similar to the theorems \ref{thm:2} and \ref{thm:3} cannot be stated in general. The reason for this is that the number of possible "types" of possibilistic paradoxes like (\ref{eq:18}) rapidly increases as the complexity of the scenario increases. For instance, in a tripartite Bell scenarios, a "brute-force calculation" shows that there exist, in addition to those considered in Eq. (\ref{eq:36}), all possible variations of the PP defined in Ref. \cite{WM12}.  Therefore, unlike the simple scenarios, where there is a friendly and closed form the possibilistic paradoxes whose occurrence is necessary and sufficient for logical contextuality, for more general scenarios there may be quite a number of different types. Another possible explanation can be founded in Ref. \cite{Simmons18}.

The fact that noncontextual assignments to observable properties lead to contradictions with quantum predictions is a well-known fact since the seminal paper by Kochen-Specker (KS) \cite{KS67}. In particular, KS demonstrated the existence of a finite set of projectors in $\mathbb{R}^3$ that are not compatible with noncontextual assignment of values (true/false). The first step in the proof consists of identifying a set of eight vectors whose
orthogonality relations are represented by the “KS-bug” graph. This construction has a peculiar property known as true-implies-false (TIF): there exist vertices $A$ and $B$ such that whenever $A$ is true, then $B$ must be false. In addition of being at the basis of several KS-type contradictions \cite{BCGKL}, "KS-bug" is the simplest example of TIF structure \cite{BPSS18}. There is a strong connection between these structures and inequality-free proofs, since the vectors used in the Hardy's inequality-free proof \cite{Hardy93} appear also in the famous "18 vectors proof" proposed in Ref. \cite{CSA96}, and the aforementioned vectors appearing in the Cabello \textit{et al.} inequality-proof \cite{CBTCB13}, together with the state vector and more two additional vectors, define a KS-bug.

Finally, another interesting problem that can be explored is the "Hardy-Tsirelson problem". This problem consists in finding the maximum value allowed by quantum theory for the non-zero probability in possibilistic paradoxes. We know, for instance, that $\gamma_4=(5\sqrt{5}-11)/2$ is the Hardy-Tsirelson bound for the CHSH scenario \cite{RZS12}. The behavior of the $\gamma_n$ in Fig. \ref{fig:8}, however, suggests that it does not correspond to the Tsirelson bound value for all $n$.

\begin{acknowledgments}
The authors thank the members of Quantum Foundations group (IF-USP) for the discussions and comments; Dr. Jader for his help with Mathematica; and Profs. C. Budroni, E. Dzhafarov, M. Terra Cunha, S. Abramsky, and the anonymous referees for the comments and suggestions.  This work was supported by the Conselho Nacional de Pesquisa Científica (CNPq) through Grant No. 155496/2019-0.
\end{acknowledgments}

\end{document}